\documentclass[prl,twocolumn,amsmath,amssymb,10pt,aps,longbibliography,superscriptaddress,citeautoscript,bibnotes,floatfix]{revtex4-2}

\usepackage{graphicx} 
\usepackage{dcolumn}
\usepackage{bm}
\usepackage{graphicx, color, xcolor}
\usepackage{physics}
\usepackage{comment}
\usepackage[normalem]{ulem}
\graphicspath{ {./Figures/} }
\usepackage{textcomp}
\usepackage[utf8]{inputenc}
\usepackage[T1]{fontenc}
\usepackage{amsmath}
\usepackage{amssymb}
\usepackage{amsthm}
\usepackage{siunitx}
\usepackage{physics}
\usepackage{graphicx}
\usepackage{subfigure}
\usepackage{tabularx}
\usepackage{booktabs}
\usepackage{multirow}
\usepackage[colorlinks=true,linkcolor=blue,citecolor=blue]{hyperref} 
\usepackage{soul}
\usepackage{xcolor}
\usepackage{csquotes}

\usepackage{algorithm}
\usepackage{algpseudocode}

\algtext*{EndWhile}
\algtext*{EndIf}
\algtext*{EndFor}

\newtheorem{theorem}{Theorem}

\newtheorem{lemma}{Lemma}

\begin{document}

\title{Shallow quantum circuit for generating {extremely low}-entangled approximate state designs}

\author{Wonjun Lee}
\email{wonjun1998@postech.ac.kr}
\affiliation{Department of Physics, Pohang University of Science and Technology, Pohang, 37673, Republic of Korea}
\author{Minki Hhan}
\email{minkihhan@gmail.com}
\affiliation{Department of Electrical and Computer Engineering, The University of Texas at Austin, Texas, 78712, USA}
\author{Gil Young Cho}
\email{gilyoungcho@kaist.ac.kr}
\affiliation{Department of Physics, Korea Advanced Institute of Science and Technology, Daejeon, 34141, Korea}
\affiliation{Center for Artificial Low Dimensional Electronic Systems, Institute for Basic Science, Pohang, 37673, Republic of Korea}
\author{Hyukjoon Kwon}
\email{hjkwon@kias.re.kr}
\affiliation{School of Computational Sciences, Korea Institute for Advanced Study, Seoul, 02455, South Korea}

\begin{abstract}
Random quantum states have various applications in quantum information science. {W}e discover a new ensemble of quantum states that serve as an $\epsilon$-approximate state $t$-design while possessing extremely low entanglement, magic, and coherence. These resources can reach their theoretical lower bounds, $\Omega(\log (t/\epsilon))$, which are also proven in this work. This implies that for fixed $t$ and $\epsilon$, entanglement, magic, and coherence do not scale with the system size, i.e., $O(1)$ with respect to the total number of qubits $n$. Moreover, we explicitly construct an ancilla-free shallow quantum circuit for generating such states by transforming $k$-qubit approximate state designs into $n$-qubit ones without increasing the support size. The depth of such a quantum circuit, $O(t [\log t]^3 \log n \log(1/\epsilon))$, is the most efficient among existing algorithms without ancilla qubits. A class of quantum circuits proposed in our work offers reduced cost for classical simulation of random quantum states, leading to potential applications in quantum information processing. As a concrete example, we propose classical shadow tomography using an estimator with superpositions between only two states, {from which almost all quantum states can be efficiently certified} by requiring only $O(1)$ measurements and {classical post-processing time.}
\end{abstract}

\maketitle

Random quantum states serve as a fundamental tool for understanding the typical behaviors of complex quantum systems, such as entanglement~\cite{Page1993, hayden2006} and thermalization~\cite{Srednicki1994}, and also for simulating quantum chaotic dynamics~\cite{yoshida2017, Yoshida2019, Piroli_2020}. Due to their statistical typicality, generating ensembles of random quantum states has a wide range of applications in quantum information processing, including creating secure cryptographic keys~\cite{Ji2018, ananth2022, kretschmer2023}, benchmarking quantum devices~\cite{Knill_2008, Huang_2020, Helsen2022}, demonstrating quantum supremacy~\cite{Boixo_2018, arute2019}, and performing shadow tomography~\cite{Huang_2020, elben2023}. Since generating a truly random, i.e., Haar-random, quantum state requires a quantum circuit with a depth that grows exponentially with the number of qubits~\cite{Emerson2005, Oszmaniec2024}, the need for efficient implementations of approximate versions has become increasingly important. Two primary paradigms are \textit{approximate state $t$-designs}~\cite{Andris2007}, which mimic the first $t$ moments of the Haar distribution, and {pseudorandom states}~\cite{Ji2018}, which cannot be distinguished from Haar-random states by any polynomial-time algorithm. In this direction, significant progress has recently been made in the efficient preparation of both pseudorandom states~\cite{aaronson2023quantum, Chamon2024} and approximate state $t$-designs~\cite{feng2024}, approaching their ultimate circuit depth limits~\cite{schuster2024,cui2025}.

As Haar-random states exhibit nearly maximal entanglement~\cite{Page1993,hayden2006} that scales with the number of qubits $n$, one might expect that their approximations would also exhibit high entanglement. Nevertheless, a remarkable observation has been made that a pseudorandom state can have as little as $\omega(\log n)$ entanglement~\cite{aaronson2023quantum, gu2023little}, in stark contrast to the volume-law entanglement of Haar-random states. This opens a new direction toward efficient generation of random quantum states that require significantly fewer resources, such as entanglement, coherence, magic, {and ancilla}~\cite{aaronson2023quantum,haug2023pseudorandom,gu2023little,Zhang2025}. On the other hand, the entanglement properties of approximate state $t$-designs have not yet been thoroughly explored, whereas most existing constructions based on layered random circuits~\cite{Brand_o_2016,Haferkamp2022randomquantum,Harrow2023} or chaotic dynamics~\cite{Ho_2022,Cotler2023} typically produce highly entangled states. Hence, it remains an open question how low the entanglement and other quantum resources like magic and coherence of an approximate state $t$-design can be, and whether such states can be efficiently generated using low-depth quantum circuits. This question is not merely an abstract mathematical one, as the low entanglement would allow for the efficient manipulation and simulation of the random quantum states on both classical and quantum computers. 

\begin{figure*}[t]
    \includegraphics[width=\linewidth]{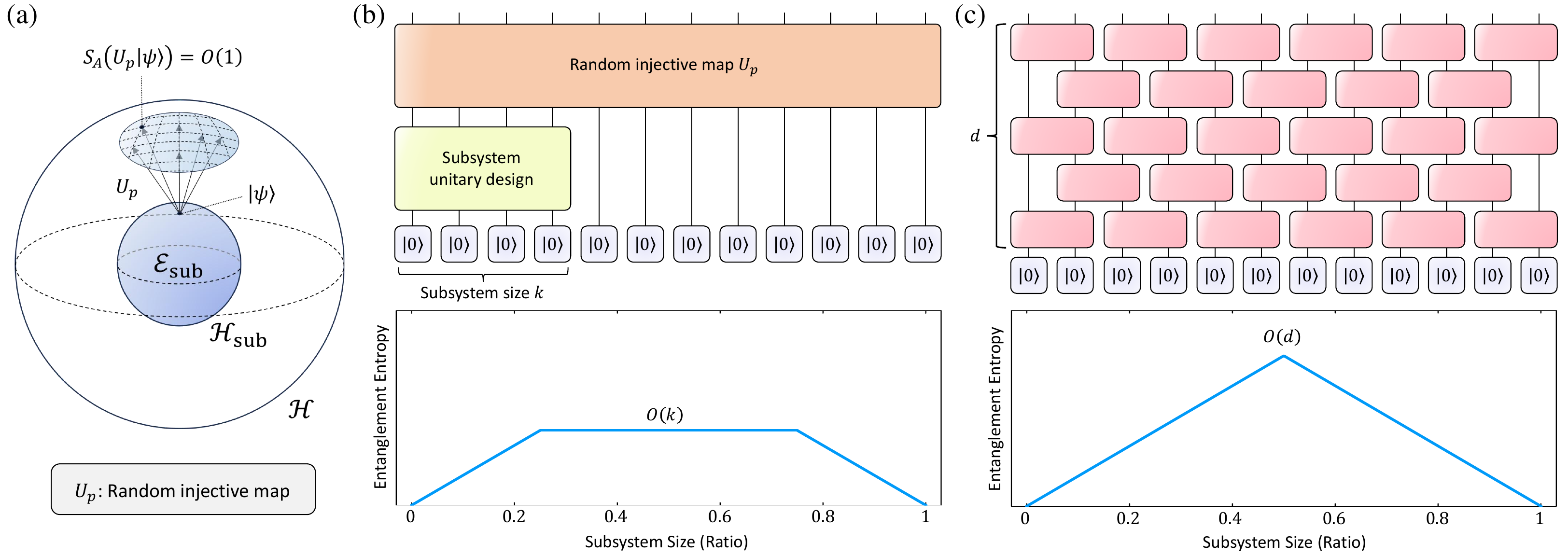}
    \centering
    \caption{Overview of the generation of approximate state designs. (a) States in an ensemble $\mathcal{E}_\mathrm{sub}$ forming an approximated state $t$-design in the subspace $\mathcal{H}_\mathrm{sub}$ are mapped to states in an ensemble forming an approximate state $t$-design in the full space $\mathcal{H}$ by random injective maps $\{U_p\}$. $S_A(\ket{\phi})$ is the half-system entanglement entropy of $\ket{\phi}$. (b,c) Generating an approximate state $t$-design using (b) {subsystem unitary design with} random injective map, and (c) a random circuit. Generated states have (b) {$O(k)$ with the subsystem size $k$} and (c) $O(d)$ entanglements with $d=\Omega(\log n)$ and the system size $n$.}
    \label{fig:overview}
\end{figure*}

In this Letter, we introduce a new class of efficiently preparable $\epsilon$-approximate quantum state $t$-designs with non-extensive, i.e., $O(1)$, entanglement, magic, and coherence for constant $\epsilon$ and $t$. In particular, we prove that those resources of these states saturate their theoretical lower bounds of $\Omega(\log(t/\epsilon))$, which we also uncover in this work. Additionally, we construct a shallow quantum circuit that can generate these quantum states efficiently. This quantum circuit is based on our new algorithm, which maps random quantum states in a $k$-qubit subsystem to the full $n$-qubit system while preserving the entanglement structure and the $t$-design property, inspired by random automaton circuits in Ref.~\cite{feng2024}. We provide an explicit $\log n$-depth quantum circuit for implementing this map based on iterations of multi-controlled NOT (MCX) gates. The overall circuit depth{, measured in basic single-and two-qubit gates,} for generating the $\epsilon$-approximate state $t$-design is $O\left(t[\log t]^3 \log n \log(1/\epsilon)\right)$ without ancilla qubits, and $\tilde{O}\left(t[\log t]^2 \log n \log(1/\epsilon)\right)$ up to polylogarithmic factors in $\log t$, when using {$\lceil n / \lceil 3\log_2 (t^2/\epsilon) \rceil \rceil$} ancilla qubits. This shallow quantum circuit uses a small number of entangling and non-Clifford gates compared to conventional random circuits, offering advantages for practical applications. In particular, we will show that such quantum circuits can dramatically improve the {classical postprocessing time} of shadow tomography~\cite{Huang_2020} as the shadow estimator contains superpositions between only two computational bases.

\textit{Results.}---Our main contributions are: {(i)} finding an ensemble of states with constant entanglement, magic, and coherence that forms an approximate state design, {(ii)} showing that the resources of these states saturates their theoretical lower bounds, {(iii)} providing an explicit shallow-circuit construction for these states, and (iv)
{designing} a new efficient classical shadow method using this circuit{, which is especially useful for quantum state certification}. Our construction is based on the intuition that the Hilbert space can be approximately covered by states with constant entanglement while typical states in the space have $\Theta(n)$ resources with the system size $n$, as depicted in Fig.\ref{fig:overview}(a). Our circuit, illustrated in Fig.~\ref{fig:overview}(b), generates only a constant amount of entanglement for any input computational basis state. In contrast, conventional random circuit constructions as in Fig.~\ref{fig:overview}(c) produce entanglement that increases with the system size~\cite{Brand_o_2016,Haferkamp2022randomquantum,Harrow2023}. 

\textit{{Extremely low}-entangled state designs.}---{The key idea behind our main results is that a $k$-qubit state design can be extended to an $n$-qubit one by applying a random unitary operator. Specifically, the random unitary operator $U_p$ is characterized by a permutation $p$ such that $U_p \left(\ket{b} \otimes \ket{0^{n-k}}\right) = \ket{p(b)}$, where $p$ is sampled from a $t$-wise independent random injective map $\mathcal{P}$~\cite{Alon2007, chen2024} (See the End Matter for its explicit definition) as follows:}
{
\begin{lemma} [Dimension expansion of state design]\label{thm:sub-to-global-design}
For a $k$-qubit $\epsilon$-approximate state $t$-design ${\cal E}_{\rm sub} =\{ \ket{\psi^{(k)}} \}$, the ensemble $\{ U_p \ket{\psi^{(k)}} \otimes \ket{0^{n-k}} \}_{\psi^{(k)} \sim {\cal E}_{\rm sub},  p \sim \mathcal{P}}$ with a $t$-wise independent random injective map $p \sim \mathcal{P}$ from $[2^k]$ to $[2^n]$ forms a $n$-qubit $\epsilon'$-approximate state $t$-design, where $\epsilon' = \epsilon + \frac{t^2}{2^{k-1}} + \frac{t^2}{2^{n-1}}$.
\end{lemma}}
{Since $U_p$ preserves the number of superpositions in the computational basis, using this lemma, we can construct approximate state designs having extremely low entanglement by applying $U_p$ to an exact design within the subspace.}
\begin{theorem}\label{thm:const-entanged-design}
    There exists an $n$-qubit $\epsilon$-approximate state $t$-design with $O(\log(t/\epsilon))$ entanglement, magic, and coherence {for any $\epsilon\geq \frac{4t^2}{2^n}$}.
\end{theorem}
\begin{proof}
{We set the subsystem size $k\leq n$ as an integer greater than $\log_2(4t^2/\epsilon)$. Let us take ${\cal P}$ as an ensemble of $t$-wise independent random injective map from $[2^k]$ to $[2^n]$. Let $\mathcal{E}_\mathrm{sub}$ be an ensemble forming an exact state $t$-design in the subspace. From \textbf{Lemma}~\ref{thm:sub-to-global-design}, $\mathcal{E}_{\rm tot}=\{U_p \ket{\psi} \}_{p \sim \mathcal{P}, \psi\sim\mathcal{E}_\mathrm{sub}}$ forms $\epsilon'$-approximate state $t$-design with $\epsilon'=\frac{t^2}{2^{k-1}}+\frac{t^2}{2^{n-1}}\leq\frac{4t^2}{2^k}$. Thus, for any given $\epsilon \geq \frac{4t^2}{2^n}$, there exists $k\leq n$ such that $\mathcal{E}_{\rm tot}$ forms an $\epsilon$-approximate state $t$-design.}

{We note that any states in $\mathcal{E}_{\rm tot}$} have $2^k$ superpositions in the computational basis, {the same as in $\mathcal{E}_{\rm sub}$}. Thus, these states have entanglement across any bipartition of the subsystems and coherence at most $O(k) = O(\log(t/\epsilon))$. Additionally, they also have magic at most $O(\log(t/\epsilon))$ as discussed in SM~\cite{Supple}. Here, we use the stabilizer R\'enyi entropy~\cite{Leone2022} and the relative entropy of coherence~\cite{Baumgratz2014} as measures of magic and coherence, respectively.
\end{proof}

\textit{Entanglement lower bound.}---We present a theoretical lower bound on the entanglement of states forming an approximate state $t$-design. These states have entanglement as well as magic and coherence {of} at least $\Omega(\log(t/\epsilon))$ due to the following theorem. We note that the states provided in \textbf{Theorem}~\ref{thm:const-entanged-design} saturate this lower bound.
\begin{theorem}\label{thm:ent_lower_bound}
    An ensemble of pure quantum states can form an $\epsilon$-approximate state $t$-design {only if} mean entanglement, magic, and coherence of the states are at least $\Omega(\log(t/\epsilon))$.
\end{theorem}
{A proof of \textbf{Theorem}~\ref{thm:ent_lower_bound} is provided in SM~\cite{Supple}. \textbf{Theorem}~\ref{thm:ent_lower_bound}, together with \textbf{Theorem}~\ref{thm:const-entanged-design}, establishes fundamental tight lower bounds $\Theta(\log(t/\epsilon))$ on entanglement, magic, and coherence of $\epsilon$-approximate state $t$-designs. This implies, surprisingly, that for fixed $\epsilon$ and $t$, these resources can become constants, independent of the number of qubits. Such an ensemble of states can be utilized to estimate the average fidelity of a given quantum channel $\Lambda$, ${\cal F}_{\rm avg}(\Lambda) = \mathbb{E}_{\psi\sim{\rm Haar}}\left[ \bra{\psi} \Lambda(\ket{\psi}\bra{\psi}) \ket{\psi} \right]$~\cite{Horodecki1999, Nielsen2002}, since an $\epsilon$-approximate state $2$-design can estimate ${\cal F}_{\rm avg}(\Lambda)$ within error $2\epsilon$, regardless of the system size $n$. We also highlight that pseudorandom states with the lowest possible entanglement $\omega(\log n)$~\cite{aaronson2023quantum} do not guarantee an accurate approximation of the average fidelity, as these random states are indistinguishable from Haar-random states only under limited computational time.}

{Meanwhile, for certain applications in quantum information processing, $t$ and $\epsilon$ may scale with the system size $n$. For example, in cryptographic applications, $t$ and $\epsilon$ are typically set to scale $\mathrm{poly}(n)$ and $1/\omega(\mathrm{poly}(n))$, respectively. In this case, our lower bounds become $\omega(\log n)$, which agree with known lower bounds for the resources of pseudorandom states~\cite{aaronson2023quantum,gu2023little}.}

\begin{figure}[t]
    \includegraphics[width=\linewidth]{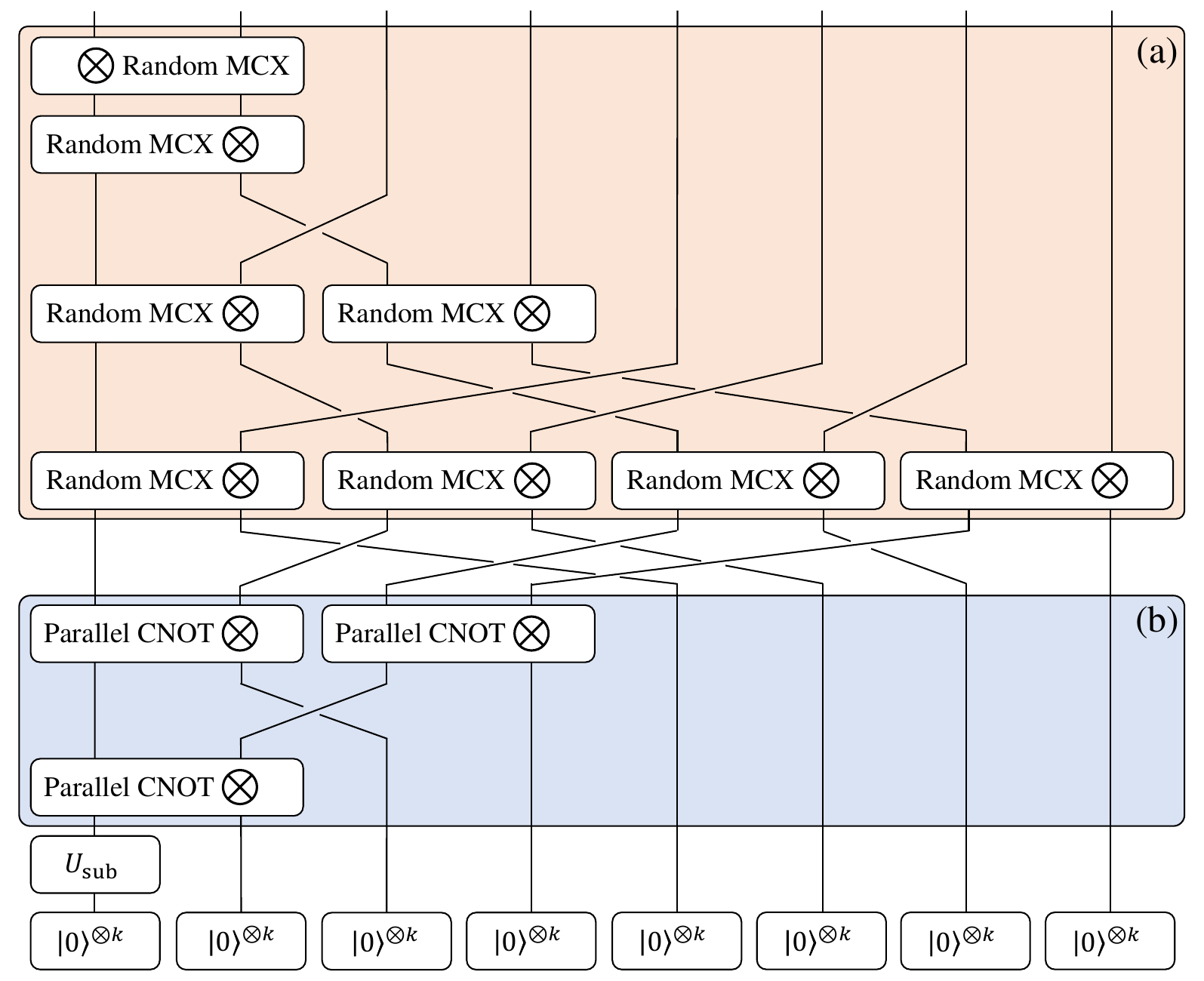}
    \centering
    \caption{Circuit for generating $O(1)$-entangled state designs. `$\otimes$' in each conditional gate represents its target bit. {$U_\mathrm{sub}$ represents a circuit generating state designs in a subsystem.} The circuit $U_p$ for a random injective map can be implemented by (a) a random MCX gate circuit and (b) a parallel CNOT gate circuit.}
    \label{fig:circuit}
\end{figure}

\textit{Shallow circuit implementation.}---{To translate our findings into practice, we develop} an algorithm {implementing $U_p$} in \textbf{Theorem}~\ref{thm:const-entanged-design} using a shallow circuit. $U_p$ can be implemented by parallel Control NOT (CNOT) gates and random MCX gates, illustrated in Fig.~\ref{fig:circuit}. More details on the implementation, including two main algorithms (\textbf{Algorithms}~\ref{alg:depth_opt} and \ref{alg:RMCC}), can be found in End Matter. {We also provide a more efficient algorithm for implementing $U_p$ for $t\leq 3$ in SM~\cite{Supple}.} The following theorem guarantees that our algorithms indeed output circuits generating an $O(1)$-entangled state $t$-design in $O(\log n)$ depth.
\begin{theorem}\label{thm:approximate-design-depth-main}
    The states in \textbf{Theorem}~\ref{thm:const-entanged-design} can be implemented in $O(t [\log t]^3 \log n \log (1/\epsilon))$ depth without ancilla qubits and $\tilde{O}(t [\log t]^2 \log n \log (1/\epsilon))$ depth with {$\lceil n/\lceil 3\log_2(t^2/\epsilon)\rceil\rceil$} ancilla qubits where $\tilde{O}(\cdot)$ neglects $\log\log t$ terms. {Here, the circuit depth is counted by using single-qubit gates and CNOT gates with all-to-all connectivity.}
\end{theorem}
\noindent
We prove this theorem in SM~\cite{Supple}. {We compare the circuit depth of our algorithm with previous approaches for generating approximate state designs without ancillary qubits in Table~\ref{tab:depth}. We highlight that our method becomes the most efficient among these approaches. By taking $n = 58, t = 3, \epsilon = 0.01$, as an explicit example, our construction requires {58-depth}, which could be achievable in near-term quantum devices using parallel applications of entangling gates with all-to-all connectivity~\cite{Evered2023,Xue2024,Bluvstein2024}.

When using ancillary qubits, recent works have shown that $\epsilon$-approximate unitary $t$-designs can be constructed with $O(\log t\log\log (nt/\epsilon))$ depth~\cite{schuster2024, cui2025}, although the resulting states exhibit at least logarithmic entanglement scaling in $n$. This leaves open the question of whether the gluing approach used in Refs.~\cite{schuster2024, cui2025} can also lead to improvements in generating state $t$-designs with low entanglement, magic, and coherence.}

\begin{table}[t] 
    {
    \begin{ruledtabular} 
        \begin{tabular}{ccc}
            Depth (time) & Error type & Method \\
            \hline
            $\mathrm{poly}(n)\cdot n^{1/D}$ & Relative & $\begin{array}{c}
            \text{Local random circuits} \\ \text{\cite{Brand_o_2016, Haferkamp2022randomquantum, Harrow2023}} \end{array}$  \\
            $[f(t)]^{1/\gamma}\cdot O(\mathrm{poly}(n))^\star$ & Additive & 
            Projected ensembles~\cite{Cotler2023}\\
            $O(nt^2 \log n)$ & Additive & Automaton circuits \cite{feng2024} \\
            $O(t [\log t]^7 \log n)$ & Relative & Gluing small designs \cite{schuster2024}\\
            $O(t[\log t]^3 \log n)$ & Additive & $\begin{array}{c} \text{Expanding design} \\ \text{(This work)} \end{array}$
        \end{tabular}
        \null\hfill $^\star$ $f(t)$ is some increasing function and $\gamma$ is a constant.
    \end{ruledtabular}
    \caption{Comparison between various protocols for generating an approximate state $t$-design without ancilla qubits. Here, the approximation error $\epsilon$ is fixed as a constant.}
    \label{tab:depth}
    }
\end{table}

{
\textit{Applications in shadow tomograph and state certification}.---Based on the approach of expanding randomness to a larger dimensional system, we construct a novel method for classical shadow tomography~\cite{Huang_2020}, in which the shadow estimator contains only $O(1)$ superpositions in the computational basis. This enables the unbiased estimation of many observables of a complex quantum state with significantly reduced classical post-processing time.

More precisely, we construct a shadow estimator of $\rho$ by separately estimating diagonal ($\hat\rho_{\rm d.}$) and off-diagonal ($\hat\rho_{\rm o.d.}$) elements as follows:
\begin{equation}\label{eq:shadow_estimator}
\begin{aligned}
    \hat\rho_{\rm d.} &= \ketbra{z}, &&z\sim\bra z \rho \ket z\\ 
    \hat\rho_{\rm o.d.} &= \frac{(2^k+1)(2^n-1)}{2^k-1} U^\dag\ketbra{z}U,&& z\sim\bra{z}U\rho U^\dag\ket{z}.
\end{aligned}
\end{equation}
Here, $z\in\{0,1\}^n$ are measurement outcomes and the random unitary $U$ for estimating the off-diagonal elements is given as
\begin{equation}\label{eq:shadow-circuit}
    U=(V\otimes I^{\otimes(n-k)})U_p^\dag,
\end{equation}
where $V$ is sampled from an $k$-qubit unitary 3-design. We provide a slight variant construction of $U_p$ that is optimized for shadow tomography in SM~\cite{Supple}. By decomposing $\hat{O}$ into the diagonal and off-diagonal components as $\hat{O} = \hat O_{\rm d.} + \hat O_{\rm o.d.}$, both $\hat{o}_\mathrm{d.} = \tr(\hat{\rho}_\mathrm{d.}\hat{O}_\mathrm{d.})$ and $ \hat{o}_\mathrm{o.d.} = \tr(\hat{\rho}_\mathrm{o.d.}\hat{O}_\mathrm{o.d.})$ can be unbiasedly estimated by the following theorem:
\begin{theorem}
    For any operator $\hat O $, $\hat{o}_\mathrm{d.}$ and $\hat{o}_\mathrm{o.d.}$ give an unbiased estimation as 
    \begin{equation}
        \tr(\rho\hat{O}) =  \mathbb{E}_{z\sim \bra z \rho \ket z}\left[ \hat{o}_\mathrm{d.} \right] + \mathbb{E}_{U \sim {\cal E}, z\sim\bra{z}U\rho U^\dag\ket{z}} \left[ \hat{o}_\mathrm{o.d.} \right].
    \end{equation}
    Furthermore, their variances are upper-bounded by
    \begin{equation}\label{eq:shadow_variance}
        \begin{split}
            \mathrm{Var}(\hat{o}_\mathrm{d.})&\leq \tr(\hat{O}^2)\\
            \mathrm{Var}(\hat{o}_\mathrm{o.d.})&\leq O\left(\tr(\hat{O}^2)\right) + O\left(2^{n-k}\chi_\rho(\hat{O})\right)
        \end{split}
    \end{equation}
    with $\chi_\rho(\hat{O})=\sum_{z\in[2^n]}\bra{z}\rho\ket{z}\bra{z}\hat{O}^2_{\rm o.d.}\ket{z}$. Consequently, the shadow norm that determines the sample complexity~\cite{Huang_2020} is given by
    \begin{equation}
        \|\hat{O}\|_\mathrm{shadow}^2 \leq O\left(\tr(\hat{O}^2)\right) + O\left(2^{n-k}\chi(\hat{O})\right) 
    \end{equation}
    with $\chi(\hat{O})=\displaystyle\max_{z \in \{0,1\}^n} \left[ \bra z \hat O^2 \ket {z} - \bra{z} \hat O \ket {z}^2 \right]$.
\end{theorem}
This theorem is proven in SM~\cite{Supple}. The main advantage of our protocol over the original classical shadow tomography based on random Clifford circuits~\cite{Huang_2020} is its low classical computational cost for evaluating observables from the estimator. Since the estimated observable $\hat{O}$ must be stored and accessed classically, the classical shadow protocol is limited by the efficiency of the classical representation of $\hat{O}$. For shadow tomography of low-rank observables, a typical approach for this is using low-rank stabilizer states~\cite{Bravyi2019}, neural quantum states (NQS)~\cite{huang2024certifying,Lange_2024}, or matrix product states (MPS)~\cite{Schollwock2011} to efficiently represent the observables as their summations. These methods allow the efficient extraction of their matrix elements in the computational basis. For the random Clifford circuits, evaluating the expectation value, $\mathrm{tr} \left(\hat{\rho}_{\rm Cl} \hat{O}\right) \propto \mathrm{tr} \left(U^\dagger_{\rm Cl} \ketbra{z} U_{\rm Cl} \hat{O}\right)$, from the estimator $\hat\rho_{\rm Cl}$ for a Clifford unitary $U_{\rm Cl}$ with outcome $z$, can be computationally challenging when $\hat{O}$ has a high stabilizer rank. In contrast, our estimator requires at most $4^{k}$ access to $\hat O_{z_1 z_2} = \bra {z_1} \hat O \ket{z_2}$, which corresponds to the number of $\ketbra{z_2}{z_1}$ terms in our estimator $\hat \rho_{\rm d.}$ or $\hat \rho_{\rm o.d.}$. For the minimum case of $k=1$, only a single access to $\hat O_{z_1 z_2}$ is sufficient to evaluate the expectation value of a Hermitian observable for each sample.}

{At the same time, by taking $k = n - \log_2 \chi(\hat{O})$, one can obtain the same sample complexity scaling as that using random Clifford circuits, $3\tr(\hat{O}^2)$. In particular, for $\chi(\hat{O}) = O(n 2^{-n})$, for example, when $\hat O$'s eigenstates are the Haar random states, the classical post-processing cost remains polynomial since $k = O(\log n)$. Even for the case with large $\chi(\hat{O})$, one can apply additional unitary operator $W$ to obtain $\hat O' = W \hat O W^\dagger$ with lower $\chi( \hat O')$, and then estimate $\hat O'$ for $\rho' = W \rho W^\dagger$ by noting that $\tr \left( \rho \hat O \right) = \tr \left(\rho' \hat O'\right)$. This allows us to estimate the fidelity of product states or the GHZ state with constant sample complexity even with $k=1$ (see SM~\cite{Supple} for more detail). Moreover, from Eq.~\eqref{eq:shadow_variance}, if $\rho$ is sufficiently random in the computational basis, then typical observables with random eigenstates can also be estimated with constant sample complexity.}

{These make our formalism able to estimate a large family of observables and are particularly useful for quantum state certification~\cite{huang2024certifying,gupta2025} and benchmarking~\cite{Knill_2008,elben2023}. The state certification task aims to determine whether an experimentally generated state $\rho$ has sufficient overlap between a target state $\ket\psi$, \textit{i.e.}, $\bra\psi \rho \ket\psi\geq 1-\varepsilon$. We find that for any given $n$-qubit state $\rho$, almost all $n$-qubit pure states $\ket{\psi}$ can be certified with a constant sample complexity as follows:}
{
\begin{theorem}
    Almost all $n$-qubit pure states can be certified through $O(1/\varepsilon^2)$ measurements with $O(\log (n/\varepsilon))$-depth circuits and queries to fixed local basis coefficients of the target state.
\end{theorem}}
The proof of the theorem can be found in SM~\cite{Supple}. This provides a substantial reduction from the $O(n^2)$ {single qubit measurements} needed by recently developed methods~\cite{huang2024certifying,gupta2025}, with an added logarithmic circuit depth overhead. {Compared to Ref.~\cite{huang2024certifying}, which contains certain undecidable cases, our approach can directly estimate the fidelity under a relatively simple strict anti-concentration condition, i.e., all measurement probabilities are upper bounded by $O(2^{-n})$, for certifiable states. While Ref.~\cite{gupta2025} proposes a general certification protocol applicable to arbitrary quantum states, our method does not require direct access to the Pauli basis, as it operates directly on the NQS ansatz, which has better expressivity than the MPS ansatz. Fidelity can also be estimated via classical shadow tomography using approximate unitary designs from one-dimensional $O(\log n \log(1/\varepsilon))$-depth random circuit~\cite{schuster2024} together with the MPS ansatz, However, our method is substantially more efficient, reducing classical overhead by a factor of $n^{O(1)}$ and eliminating the $\log(1/\varepsilon)$ factor in the measurement-circuit depth.}

\textit{Discussion.}---We have provided a construction of {an $\epsilon$-approximate state $t$-design with $O(1)$-entangled states for constant $\epsilon$ and $t$}. Our construction saturates the theoretical lower bounds of mean entanglement, magic, and coherence of quantum states that form an $\epsilon$-approximate state $t$-design, which are $\Omega(\log(t/\epsilon))$. We have also developed shallow random circuits that generate such states. This directly implies that our algorithms can be used to simulate Haar random states, which have been widely used in various applications, {such as} benchmarking quantum circuits and quantum learning.

{Furthermore,} we show that our {circuit} construction can be used for classical shadow tomography with {enhanced classical post-processing time.} This feature can be utilized to {efficiently} perform a certification task with {fewer samples than previously proposed schemes and constant computation time by adapting} neural quantum states or matrix product states. We believe our work opens up the possibility of performing these tasks on near-term quantum devices. 

\textit{Note added.} {After the completion of this work, we became aware of another independent work that introduces a constant magical ensemble that forms an approximate state design~\cite{bittel2025}.}

\begin{acknowledgements}
We thank Byungmin Kang for helpful discussions. W.L. and G.Y.C. are financially supported by Samsung Science and Technology Foundation under Project Number SSTF-BA2401-03, the NRF of Korea (Grants No. RS-2023-00208291, RS-2024-00410027, 2023M3K5A1094810, RS-2023-NR119931, RS-2024-00444725, RS-2023-00256050, IRS-2025-25453111, RS-2025-08542968) funded by the Korean Government (MSIT), the Air Force Office of Scientific Research under Award No. FA23862514026, and Institute of Basic Science under project code IBS-R014-D1. W.L. is supported by the KAIST Jang Young Sil Fellow Program. H.K. is supported by the KIAS Individual Grant No. CG085302 at Korea Institute for Advanced Study and National Research Foundation of Korea (Grants No.~RS-2023-NR119931, No.~RS-2024-00413957, and No.~RS-2024-00438415) funded by the Korean Government (MSIT).
M.H. is supported by Schmidt Sciences Polymath award to David Soloveichik. 
\end{acknowledgements}

\bibliography{Ref}

\clearpage

\noindent{\large{\bf End Matter}}\\

\noindent{\bf Approximate state designs.}\\
An ensemble $\mathcal{E}$ of $n$-qubit states forms an approximate state $t$-design with additive error $\epsilon$ if it approximates Haar random states in the following sense:
\begin{equation}
\operatorname{TD}\left(\mathbb{E}_{\phi\sim\mathcal{E}}\left[\ketbra{\phi}^{\otimes t}\right],\rho_\mathrm{Haar}^{(t)}\right)\leq \epsilon
\end{equation}
with $\rho_\mathrm{Haar}^{(t)} = \mathbb{E}_{\psi\sim\mathrm{Haar}}\left[\ketbra{\psi}{\psi}^{\otimes t}\right]$. Here, $\mathbb{E}_{\phi\sim\mathcal{E}}\left[\ketbra{\phi}^{\otimes t}\right]$ is the ensemble average of $t$-copy quantum states over $\mathcal{E}$, and $\operatorname{TD}(\rho,\sigma)$ is the trace distance between two quantum states $\rho$ and $\sigma$.
\vspace{1em}

\noindent{{\bf $t$-wise independent random injective map and its unitary implementation.}\\
Here, we define a $t$-wise independent random injective function. We set $k$ and $n$ to be integers such that $k\leq n$. Let $\mathcal{P}$ be an ensemble of functions from $\{0,1\}^k$ to $\{0,1\}^n$. Then, $\mathcal{P}$ is $t$-wise independent if any $t$ distinct bitstrings $\{x_i\}_{i=1}^t$ of length $k$ and $\{y_i\}_{i=1}^t$ of length $n$, the probability of having $p(x_i)=y_i$ for all $i\in[1,t]$ is uniformly given as
\begin{equation}
    \mathrm{Pr}_{p\sim \mathcal{P}}[p(x_1)=y_1,\cdots,p(x_t)=y_t] = \frac{1}{t!}\binom{2^n}{t}^{-1}.
\end{equation}
Consequently, an ensemble of $n$-qubit unitary operators $\{U_p\}_{p\sim\mathcal{P}}$ implements $\mathcal{P}$ if it follows that
\begin{equation}
    \begin{split}
        \mathrm{Pr}_{p\sim \mathcal{P}}[U_p\ket{x_1,a}=\ket{y_1},\cdots,U_p\ket{x_t,a}=\ket{y_t}]
        = \frac{1}{t!}\binom{2^n}{t}^{-1}
    \end{split}
\end{equation}
for any $a\in\{0,1\}^{n-k}$.
}

\vspace{1em}  
\noindent{\bf Algorithm for random injective map.}\\
A random injective map $U_p$ introduced in the main text is implemented by \textbf{Algorithm}~\ref{alg:depth_opt}, which uses \textbf{Algorithm}~\ref{alg:RMCC} as a subroutine. Let us define notations used in \textbf{Algorithm}~\ref{alg:depth_opt}. For some $m$-distinct positions $s$ in $[1,k]$, a vector $c$ in $\{0,1\}^m$, and a position $q$ in $[1,n]$, $m$-$\mathrm{MCX}_{(s,c),q}$ is a MCX gate with $m$-conditional bits targeting the $q$-th qubit whose control conditions are given by $c$. $S$ and $C$ are sequences length $\alpha$ and comprised of sets of $m$-distinct random positions and random vectors in $\{0,1\}^m$, respectively. Additionally, the `$.+$' operator is the element-wise summation.\\

\begin{figure}[H]
\begin{algorithm}[H]
    \caption{Shallow Depth Bits Randomizer}
    \label{alg:depth_opt}
    \begin{algorithmic}
    \Require $n > 0$, $n\geq k > 1$, $m\geq 2$, $\alpha>0$
    \Ensure $U$
    \State $U \gets I$
    \For{$i = 1$ to $\log_2 (n/k)$; $i$\texttt{++}}
        \For{$j = 1$ to $k$; $j$\texttt{++}}
            \For{$l = 0$ to $2^{i-1}-1$; $l$\texttt{++}}\\
                \Comment{This can be done in parallel} 
                \State $U\gets \operatorname{CNOT}_{lk+j,(2^{i-1}+l)k+j} U$ 
            \EndFor
        \EndFor
    \EndFor
    \State $S, C \gets \operatorname{RMCC}(1,k,m,\alpha)$
    \For{$i = \log_2(n/k)$ to $1$; $i$\texttt{-}\texttt{-}}
        \State $T\gets$ A set of $2^{i-1}$ random bit-strings of length $\alpha$
        \For{$j=0$ to $\alpha-1$; $j$\texttt{++}}
            \For{$q=0$ to $k-1$; $q$\texttt{++}}
                \For{$l=0$ to $2^{i-1}-1$; $l$\texttt{++}}\\
                \Comment{This can be done in parallel.}
                    \If{$T[l][j]$}
                        \State $U \gets m\operatorname{-MCX}_{(lk.+S[j],C[j]),(l+2^{i-1})k+q}U$\\
                        \Comment{$lk.+S[j]$ are positions of control bits.}\\
                        \Comment{$C[j]$ are conditions of control bits.}
                    \EndIf
                \EndFor
            \EndFor
        \EndFor
    \EndFor 
    \State $T\gets$ A random bit-string of length $\alpha$
    \For{$j=0$ to $\alpha-1$; $j$\texttt{++}}
        \For{$q=0$ to $k-1$; $q$\texttt{++}}
            \If{$T[j]$}
                \State $U \gets m\operatorname{-MCX}_{(k.+S[j],C[j]),q}U$
            \EndIf
        \EndFor
    \EndFor
    \end{algorithmic}
\end{algorithm}
\end{figure}

\begin{figure}[H]
\begin{algorithm}[H]
    \caption{Random Multi-Control Condition (RMCC)}
    \label{alg:RMCC}
    \begin{algorithmic}
    \Require $x_1\geq 1$, $x_2 > x_1$, $m\geq 2$, $\alpha>0$
    \Ensure $S$, $C$
    \State $S \gets \emptyset$
    \State $C \gets \emptyset$
    \For{$i=1$ to $\alpha$; $i$\texttt{++}}
        \State $s\gets$ Choose $m$ distinct positions randomly in $[x_1,x_2]$
        \State $c\gets$ Random bits of length $m$
        \State $S\gets\operatorname{push}(S,s)$
        \State $C\gets\operatorname{push}(C,c)$
    \EndFor
    \end{algorithmic}
\end{algorithm}
\end{figure}

\end{document}



\title{Supplemental Material: Shallow quantum circuit for generating extremely low-entangled approximate state designs}
\author{Wonjun Lee}
\email{wonjun1998@postech.ac.kr}
\affiliation{Department of Physics, Pohang University of Science and Technology, Pohang, 37673, Republic of Korea}
\author{Minki Hhan}
\email{minkihhan@gmail.com}
\affiliation{Department of Electrical and Computer Engineering, The University of Texas at Austin, Texas, 78712, USA}
\author{Gil Young Cho}
\email{gilyoungcho@kaist.ac.kr}
\affiliation{Department of Physics, Korea Advanced Institute of Science and Technology, Daejeon, 34141, Korea}
\affiliation{Center for Artificial Low Dimensional Electronic Systems, Institute for Basic Science, Pohang, 37673, Republic of Korea}
\author{Hyukjoon Kwon}
\email{hjkwon@kias.re.kr}
\affiliation{School of Computational Sciences, Korea Institute for Advanced Study, Seoul, 02455, South Korea}

\maketitle
\tableofcontents

\newpage
\clearpage
\setcounter{page}{1}
\section*{Supplementary Figures}
\hypertarget{fig:copy_circuit}{}
\subsection*{Supplementary Figure 1: Parallel CNOT gate circuit}
\begin{figure}[h]
    \includegraphics[width=0.7\linewidth]{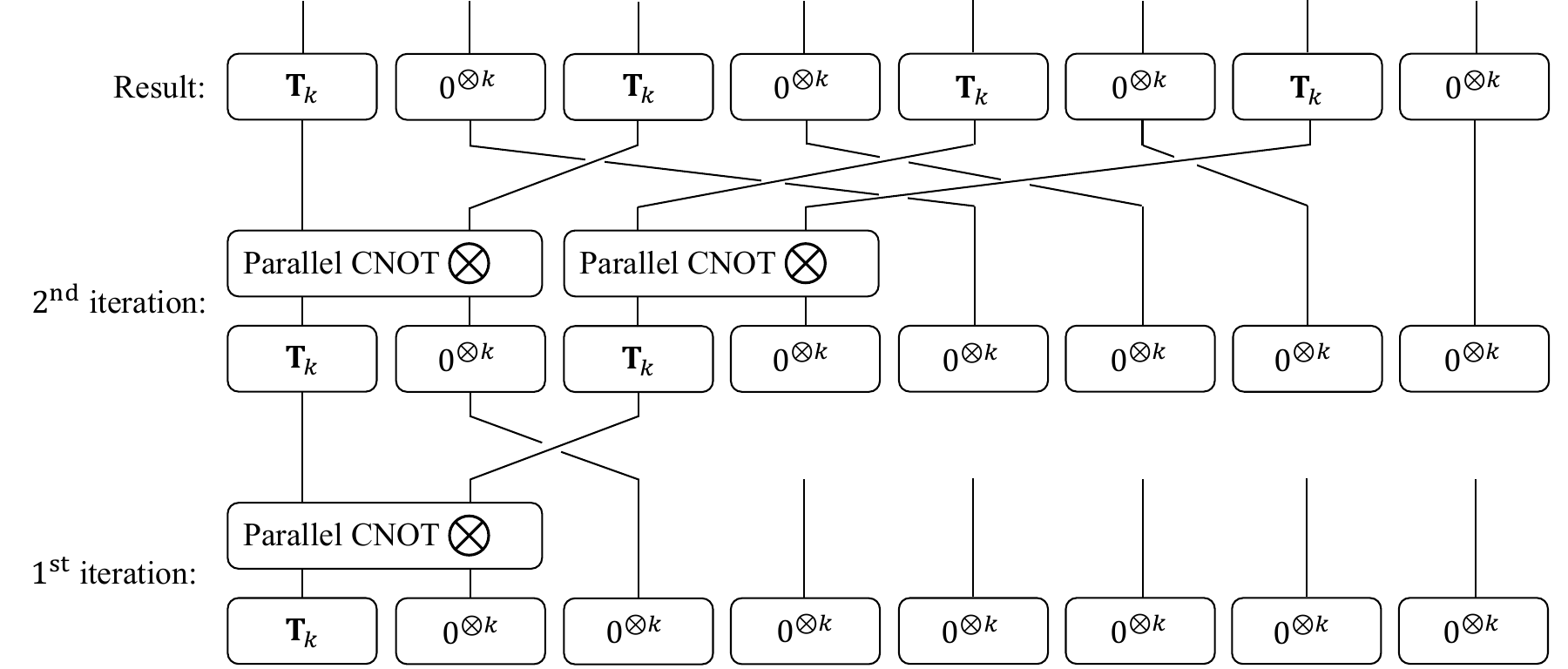}
    \begin{justify}
        Each small block at the bottom of each circuit represents a $k$-bit register. The first block is copied to half of the registers using CNOT gates with a tree-like propagation.
    \end{justify}
\end{figure}

\newpage
\hypertarget{fig:randomize_circuit}{}
\subsection*{Supplementary Figure 2: Random MCX gate circuit}
\begin{figure}[h]
    \includegraphics[width=0.7\linewidth]{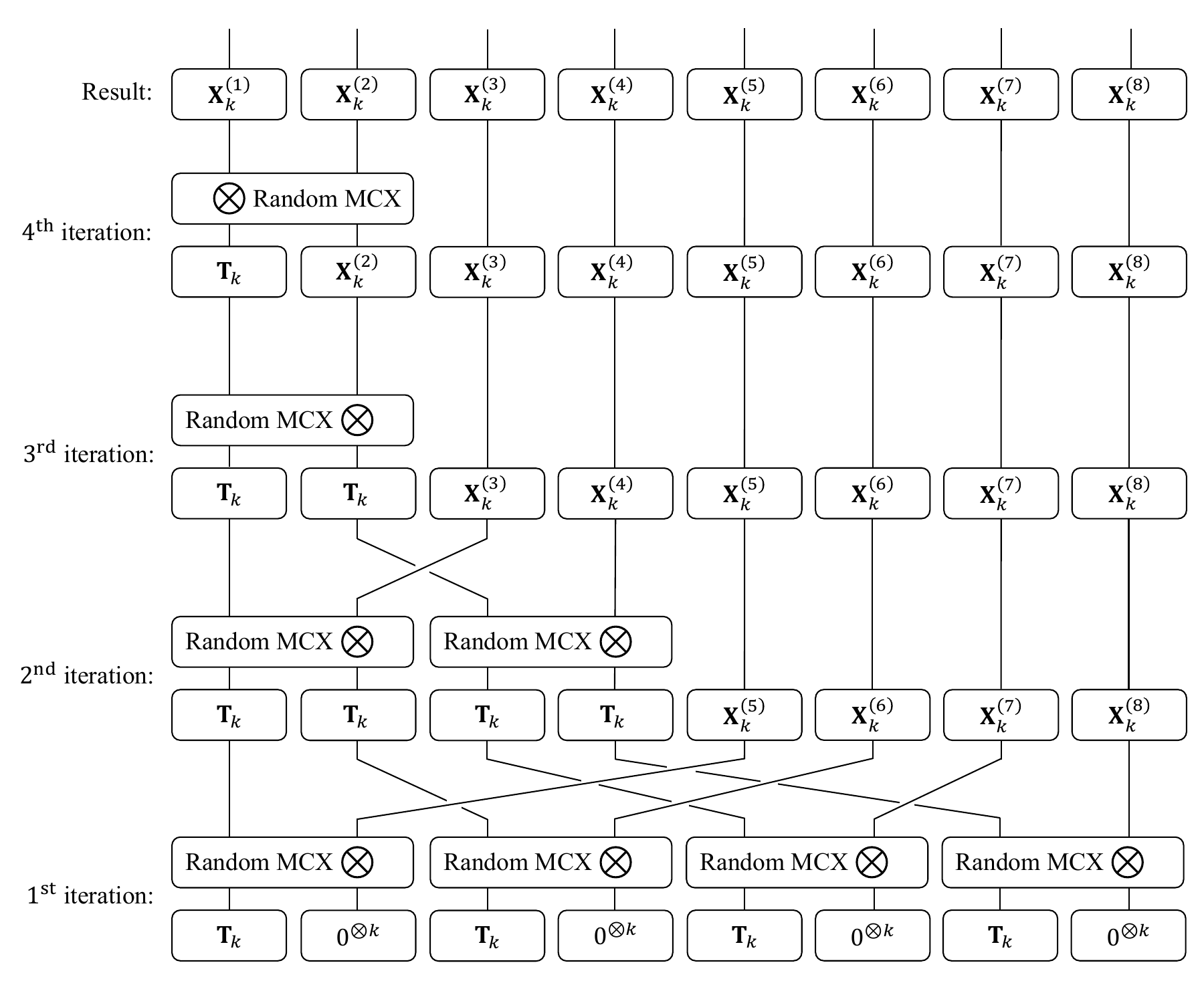}
    \begin{justify}
        Random MCX gates are applied conditioned on registers with the state $\mathbf{T}_k$ targeting bits in other registers. Random applications of MCX gates randomize targeted bits. Registers with randomized bits are denoted by $\mathbf{X}_k^{(i)}$.
    \end{justify}
\end{figure}

\newpage
\hypertarget{fig:MCX-random-var}{}
\subsection*{Supplementary Figure 3: Randomization of unique type state}
\begin{figure}[h]
    \includegraphics[width=0.6\linewidth]{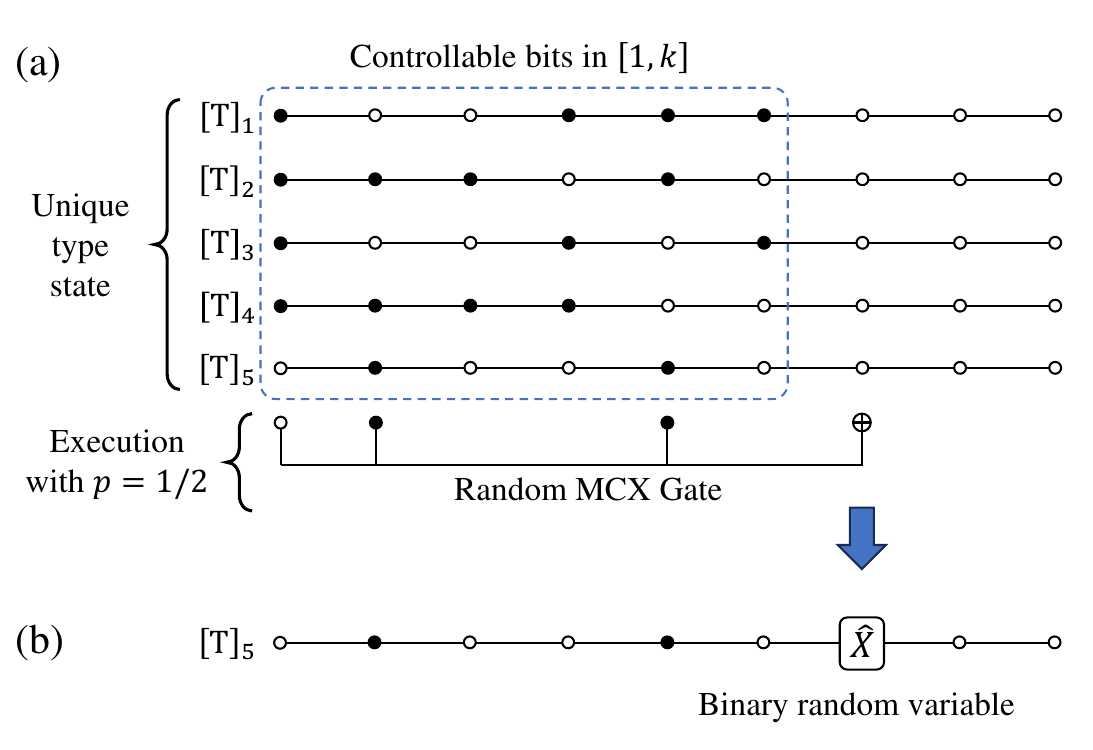}
    \begin{justify}
        (a) A 3-MCX gate sampled by the algorithm is conditioned on the first, second, and fifth bits and targeting the seventh bit. This gate is hypothetically applied on a unique type state $\ketbra{\textbf{T}}{\textbf{T}}$ of $\abs{\textbf{T}}=5$ with the probability $1/2$. (b) This random application of the MCX gate adds a binary random variable $\hat{X}$ to the value of the seventh bit of the fifth copy.
    \end{justify}
\end{figure}

\newpage
\hypertarget{fig:improved-circuit}{}
\subsection*{Supplementary Figure 4: Circuit diagrams of random multi-control X block}
\begin{figure}[h]
    \includegraphics[width=\linewidth]{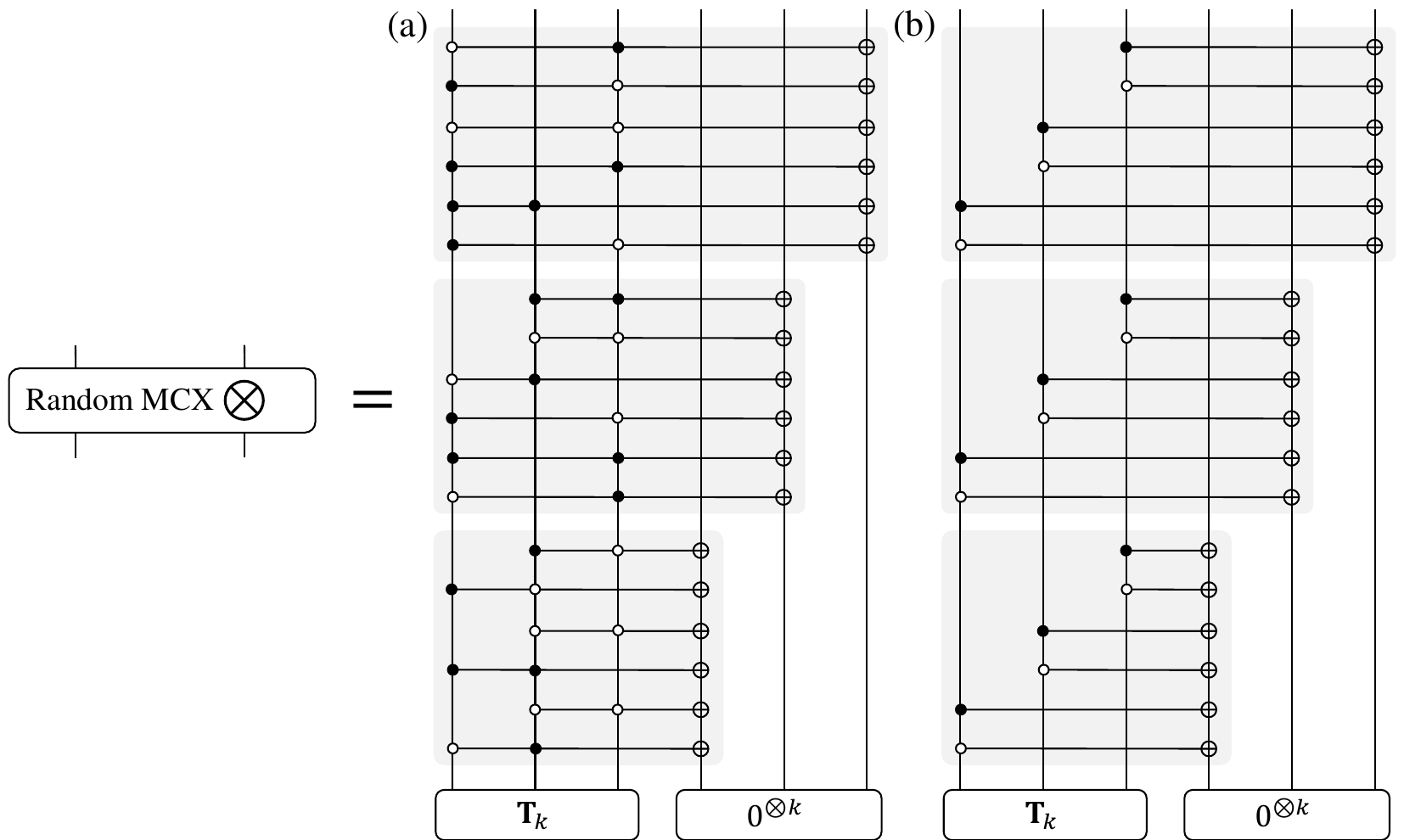}
    \begin{justify}
        Circuit diagrams of "Random MCX" block with $t=3$ for (a) our previous construction requiring $\textrm{(circuit depth for MCX gate)}\times [\textrm{(number of MCX gate applications per target qubit)} \times \textrm{(number of target qubits)} \divisionsymbol \textrm{(number of parallelizable MCX gates)}] \times \textrm{(expected number of applications)} = 6\times [3(\log 3)^2 k] \divisionsymbol 2$ mean circuit depth considering the MCX parallelization and (b) the improved implementation requiring $k$ mean circuit depth with $k=\lceil 2.885\log_2(9/\epsilon)\rceil$. This choice of $k$ is discussed in Supplementary Notes 3. Each gate is applied with the probability $p=1/2$. When $\epsilon=0.1$, the mean circuit depth of our previous construction is 207, while that of our improved construction is 19. 
    \end{justify}
\end{figure}

\newpage
\hypertarget{fig:shadow-circuit}{}
\subsection*{Supplementary Figure 5: Circuit diaram of random injective map for shadow tomography}
\begin{figure}[h]
    \includegraphics[width=0.7\linewidth]{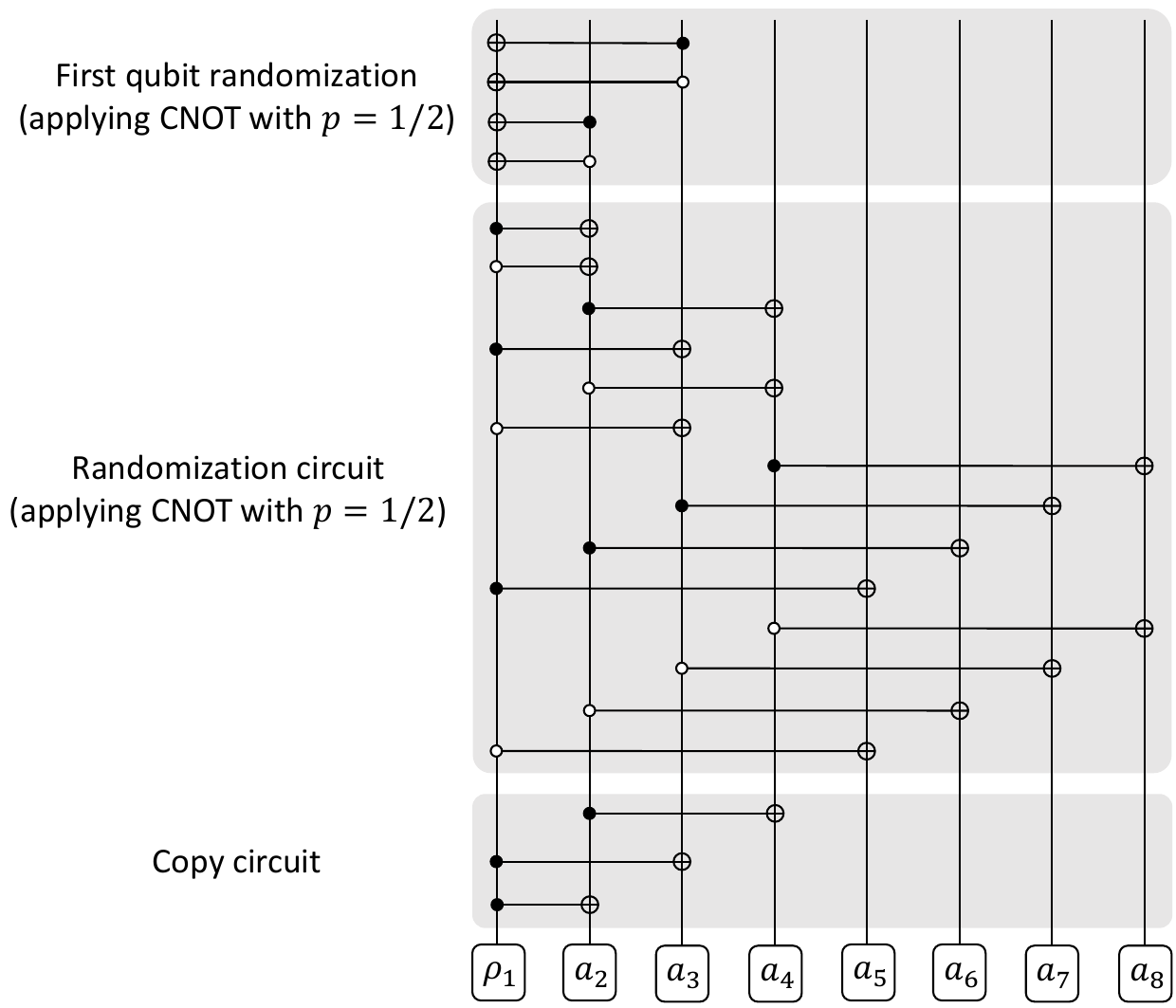}
    \begin{justify}
        The circuit diagram of random injective map for shadow tomography. This comprised of two parts: copy circuit and randomization circuit. This diagram considers two 8-bit input bitstrings, denoted as $\rho^{(1)}_1a_2\cdots a_8$ and $\rho^{(2)}_1a_2\cdots a_8$. The circuit generates two distinct random bitstrings if the leading bits $\rho^{(1)}_1$ and $\rho^{(2)}_1$ differ, or identical random bistrings if they are the same. To achieve this, the copy circuit propagates information regarding the equality of $\rho^{(1)}_1$ and $\rho^{(2)}_1$ to half of the system. Based on this information, the randomization circuit randomizes the input bitstrings in parallel.
    \end{justify}
\end{figure}



\newpage
\section*{Supplementary Notes}
\subsection*{1. Construction of state $t$-design from random injective map (proof of \textbf{Lemma 1} {in the main text})}\label{sec:state-design-construction}
In this section, we introduce the unique type state introduced in Refs.~\cite{harrow2013church, ananth2022function}, and then, prove \textbf{Lemma 1} based on it. Let us consider a vector $v \in [L]^t$ for $L=2^l$. A type vector $\operatorname{type}(v)$ of $v$ is another vector in $[t+1]^L$ whose $z$-th entry with $z\in[L]$ is the frequency of $z$ in $v$. A unique type vector $\textbf{T}_l$ is defined as a type vector whose elements are at most one. We set $\mathcal{U}_{l,t}$ as the set of all such unique type vectors. For a given unique type vector $\textbf{T}_l\in\mathcal{U}_{l,t}$, the unique type state $\ket{\textbf{T}_l}$ is defined as
\begin{equation}
    \ket{\mathbf{T}_l} = C \sum_{\substack{v\in[L]^t\\\operatorname{type}(v)=\mathbf{T}_l}}\ket{v}
\end{equation}
with a normalization constant $C$. Let $\mathcal{U}_{l,t}$ be the set of all $l$-qubit unique type states. Then, the ensemble average of unique type states is given by:
\begin{equation}\label{eq:avg-unique-type-state}
     \rho_{\mathrm{unique},l}^{(t)} = \binom{L}{t}^{-1}\sum_{\ket{\textbf{T}_l}\in\mathcal{U}_{l,t}}\ketbra{\textbf{T}_l}.
\end{equation}
Additionally, for later use, we define the ensemble average of the $t$-th moments of Haar random states in the $l$-qubit system as
\begin{equation}
    \rho^{(t)}_{\mathrm{Haar},l} = \mathbb{E}_{\psi\sim\mathrm{Haar}(2^l)}[\ketbra{\psi}^{\otimes t}].
\end{equation}

\begin{lemma}\label{thm:sub-global-unique}
    For any $k$-qubit $t$-copy unique type state $\ket{\mathbf{T}_k}$, a $t$-wise independent random injective map $p\sim\mathcal{P}$ satisfies
    \begin{equation}
        \mathbb{E}_{p\sim\mathcal{P}}[U_p^{\otimes t}\ketbra{\mathbf{T}_k}\otimes\ketbra{0}^{\otimes(n-k)t}U_p^{\otimes t,\dag}] = \rho^{(t)}_{\mathrm{unique},n}.
    \end{equation}
\end{lemma}
\begin{proof}
    We start from the definition of the unique type state:
    \begin{equation}
        \ket{\mathbf{T}_k} = \frac{1}{\sqrt{k!}}\sum_{\sigma\in S_t}\ket{(\mathbf{T}_k)_{\sigma(1)},\cdots,(\mathbf{T}_k)_{\sigma(t)}}.
    \end{equation}
    Next, let $p$ be an injective map sampled from $\mathcal{P}$. Then, we have
    \begin{equation}
        U_p^{\otimes t} \ket{\mathbf{T}_k}\otimes\ket{0}^{\otimes(n-k)t} = \frac{1}{\sqrt{k!}}\sum_{\sigma\in S_t}\ket{p((\mathbf{T}_k)_{\sigma(1)}),\cdots,p((\mathbf{T}_k)_{\sigma(t)})},
    \end{equation}
    which is a unique type state in the $n$-qubit system. Thus, the ensemble average of this over $p\sim\mathcal{P}$ is the equal summation of all possible unique type states in the $n$-qubit system, \textit{i.e.},
    \begin{equation}
        \mathbb{E}_{p\sim\mathcal{P}}[U_p^{\otimes t}\ketbra{\mathbf{T}_k}\otimes\ketbra{0}^{\otimes(n-k)t}U_p^{\otimes t,\dag}] = \rho^{(t)}_{\mathrm{unique},n}.
    \end{equation}
\end{proof}

\begin{lemma}\label{thm:Haar-unique}
    For any integers $l>0$ and $t>1$, the following is satisfied:
    \begin{equation*}
        \operatorname{TD}\left(\rho_{\mathrm{Haar},l}^{(t)},\rho^{(t)}_{\mathrm{unique},l}\right)\leq \frac{t^2}{2^{l-1}}.
    \end{equation*}
\end{lemma}
\begin{proof}
    The $t$-th moments of $l$-qubit Haar random states is explicitly given by
    \begin{equation}
        \rho_{\mathrm{Haar},l}^{(t)}=\binom{2^l+t-1}{t}^{-1}\frac{1}{t!}\sum_{\sigma\in S_t}\sum_{\{z_i\}_{i=1}^t\in[2^l]}\ketbra{\sigma(z_1,\cdots,z_t)}{z_1,\cdots,z_t}.
    \end{equation}
    Let the set of all $l$-qubit $t$-copy unique type vectors be $\mathcal{U}_{l,t}$. We introduce the projector $\Pi^{(t)}_{\mathrm{unique},l}$ as
    \begin{equation}
        \Pi^{(t)}_{\mathrm{unique},l} = \sum_{\mathbf{T}_l\in\mathcal{U}_{l,t}}\ketbra{\mathbf{T}_l}.
    \end{equation}
    We note that any $\mathbf{T}_l\in\mathcal{U}_{l,t}$, $\ket{\mathbf{T}_l}$ is an eigen state of $\rho^{(t)}_{\mathrm{Haar},l}$:
    \begin{equation}
        \begin{split}
            \rho^{(t)}_{\mathrm{Haar},l}\ket{\mathbf{T}_l}
            &= \rho^{(t)}_{\mathrm{Haar},l} \frac{1}{\sqrt{t!}}\sum_{\pi\in S_t}\ket{\pi((\mathbf{T}_l)_1,\cdots,(\mathbf{T}_l)_t)}\\
            &= \binom{2^l+t-1}{t}^{-1}\frac{1}{t!\sqrt{t!}}\sum_{\sigma,\pi\in S_t}\ket{\sigma\cdot\pi((\mathbf{T}_l)_1,\cdots,(\mathbf{T}_l)_t)} \\
            &= \binom{2^l+t-1}{t}^{-1}\frac{1}{\sqrt{t!}}\sum_{\sigma\in S_t}\ket{\sigma((\mathbf{T}_l)_1,\cdots,(\mathbf{T}_l)_t)} \\
            &= \binom{2^l+t-1}{t}^{-1}\ket{\mathbf{T}_l}.
        \end{split}
    \end{equation}
    This implies that 
    \begin{equation}
        \tr(\Pi^{(t)}_{\mathrm{unique},l}\rho^{(t)}_{\mathrm{Haar},l}) = \binom{2^l+t-1}{t}^{-1}\binom{2^l}{t}=\frac{(2^l-1)!}{(2^l+t-1)!}\frac{(2^l)!}{(2^l-t)!}.
    \end{equation}
    At the same time, $\rho^{(t)}_{\mathrm{unique},l}$ is given by
    \begin{equation}
        \rho^{(t)}_{\mathrm{unique},l} = \binom{2^l}{t}^{-1}\sum_{\mathbf{T}_l\in\mathcal{U}_{l,t}}\ketbra{\mathbf{T}_l}.
    \end{equation}
    Since $\rho_{\mathrm{Haar},l}^{(t)}$ is positive semi-definite and has unit trace, we have
    \begin{equation}
        \begin{split}
            \operatorname{TD}\left(\rho_{\mathrm{Haar},l}^{(t)},\rho^{(t)}_{\mathrm{unique},l}\right) 
            &= \operatorname{TD}\left((I-\Pi^{(t)}_{\mathrm{unique},l})\rho_{\mathrm{Haar},l}^{(t)},\rho^{(t)}_{\mathrm{unique},l}\right) + \operatorname{TD}\left(\Pi^{(t)}_{\mathrm{unique},l}\rho_{\mathrm{Haar},l}^{(t)},\rho^{(t)}_{\mathrm{unique},l}\right) \\
            &= 1-\tr(\Pi^{(t)}_{\mathrm{unique},l}\rho^{(t)}_{\mathrm{Haar},l}) + \mathrm{TD}\left(\Pi^{(t)}_{\mathrm{unique},l}\rho^{(t)}_{\mathrm{Haar},l},\rho^{(t)}_{\mathrm{unique},l}\right)\\
            &=1-\frac{(2^l-1)!}{(2^l+t-1)!}\frac{(2^l)!}{(2^l-t)!}+\left(\frac{(2^l-t)!}{(2^l)!}-\frac{(2^l-1)!}{(2^l+t-1)!}\right)\frac{(2^l)!}{(2^l-t)!}\\
            &=2\left(1-\frac{(2^l-1)!}{(2^l+t-1)!}\frac{(2^l)!}{(2^l-t)!}\right)\\
            &=2\left(1-\prod_{i=1}^t \left(1-\frac{t-1}{2^l+i-1}\right) \right)\\
            &\leq 2\cdot\sum_{i=1}^t \frac{t-1}{2^l+i-1}\\
            &\leq \frac{t^2}{2^{l-1}},
        \end{split}
    \end{equation}
{where for the first inequality, we used the fact that $\prod_i (1- a_i) \geq 1- \sum_i a_i$ for $0 \leq a_i \leq 1$}.
\end{proof}

\begin{lemma} [\textbf{Lemma 1} {in the main text}]\label{thm:supp-sub-to-global-design}
For a $k$-qubit $\epsilon$-approximate state $t$-design ${\cal E}_{\rm sub} =\{ \ket{\psi^{(k)}} \}$, the ensemble $\mathcal{E}=\{ U_p \ket{\psi^{(k)}} \otimes \ket{0^{n-k}} \}_{\psi^{(k)} \sim {\cal E}_{\rm sub},  p \sim \mathcal{P}}$ with a $t$-wise independent random injective map $p \sim \mathcal{P}$ from $[2^k]$ to $[2^n]$ forms a $n$-qubit $\epsilon'$-approximate state $t$-design, where $\epsilon' = \epsilon + \frac{t^2}{2^{k-1}} + \frac{t^2}{2^{n-1}}$.
\end{lemma}
\begin{proof}
    We denote $\rho^{(t)}_\mathrm{sub}$ and $\rho^{(t)}$ as the ensemble averages of the $t$-th moments of states in $\mathcal{E}_\mathrm{sub}$ and $\mathcal{E}$, respectively, \textit{i.e.},
    \begin{align}
        \rho^{(t)}_\mathrm{sub} &= \mathbb{E}_{\psi\sim\mathcal{E}_\mathrm{sub}}[\ketbra{\psi}^{\otimes t}]\\
        \rho^{(t)} &= \mathbb{E}_{\psi\sim\mathcal{E}}[\ketbra{\psi}^{\otimes t}].
    \end{align}
    By the definition of the approximate state design, $\rho^{(t)}_\mathrm{sub}$ approximates $\rho^{(t)}_{\mathrm{Haar},k}$ as
    \begin{equation}
        \mathrm{TD}\left(\rho^{(t)}_\mathrm{sub},\rho^{(t)}_{\mathrm{Haar},k}\right) \leq \epsilon.
    \end{equation}
    Due to \textbf{Lemma}~\ref{thm:Haar-unique}, we have
    \begin{equation}
        \mathrm{TD}\left(\rho^{(t)}_{\mathrm{Haar},k},\rho^{(t)}_{\mathrm{unique},k}\right)\leq \frac{t^2}{2^{k-1}}.
    \end{equation}
    Since $\rho^{(t)}_{\mathrm{unique},k}$ is the equal summation of all possible $k$-qubit unique type states, it follows from \textbf{Lemma}~\ref{thm:sub-global-unique} that
    \begin{equation}
        \mathbb{E}_{p\sim\mathcal{P}}[U_p^{\otimes t}\rho^{(t)}_{\mathrm{unique},k}\otimes\ketbra{0}^{\otimes(n-k)t}U_p^{\otimes t,\dag}] = \rho^{(t)}_{\mathrm{unique},n}.
    \end{equation}
    Again, due to \textbf{Lemma}~\ref{thm:Haar-unique}, $\rho^{(t)}_{\mathrm{unique},n}$ approximates $\rho^{(t)}_{\mathrm{Haar},n}$ as
    \begin{equation}
        \mathrm{TD}\left(\rho^{(t)}_{\mathrm{Haar},n},\rho^{(t)}_{\mathrm{unique},n}\right)\leq \frac{t^2}{2^{n-1}}.
    \end{equation}
    Then, the triangle inequality and the strong convexity of the trace distance imply that
    \begin{equation}
        \begin{split}
            \mathrm{TD}\left(\rho^{(t)},\rho^{(t)}_{\mathrm{Haar},n}\right)
            &\leq \mathrm{TD}\left(\rho^{(t)},\rho^{(t)}_{\mathrm{unique},n}\right) + \mathrm{TD}\left(\rho^{(t)}_{\mathrm{unique},n}, \rho^{(t)}_{\mathrm{Haar},n}\right)\\
            &\leq \mathrm{TD}\left(\rho^{(t)}_\mathrm{sub},\rho^{(t)}_{\mathrm{unique},k}\right) + \frac{t^2}{2^{n-1}}\\
            &\leq \mathrm{TD}\left(\rho^{(t)}_\mathrm{sub},\rho^{(t)}_{\mathrm{Haar},k}\right) + \mathrm{TD}\left(\rho^{(t)}_{\mathrm{Haar},k},\rho^{(t)}_{\mathrm{unique},k}\right) + \frac{t^2}{2^{n-1}}\\
            &\leq \epsilon + \frac{t^2}{2^{k-1}} + \frac{t^2}{2^{n-1}}.
        \end{split}
    \end{equation}
    Thus, $\mathcal{E}$ forms a $n$-qubit approximate state $t$-design with error
    \begin{equation}
        \epsilon' = \epsilon + \frac{t^2}{2^{k-1}} + \frac{t^2}{2^{n-1}}.
    \end{equation}
\end{proof}

\newpage
\subsection*{2. Quantum circuit implementation of random injective map (proof of \textbf{Theorem 3} in the main text)}
In this section, we prove \textbf{Theorem 3}. To this end, we first state and prove the following lemmas.

We note that $U_f$ with an exact $2t$-wise random function $f$ can be implemented in $O(kt)$ depth using circuits with ancilla bits for integer arithmetic~\cite{ALON1986,metger2024}. However, as it requires ancilla bits, we stick with the approximate construction of Ref.~\cite{chen2024}.

\begin{lemma}\label{thm:implement-U_p}
    Let $\epsilon$ and $t$ be a positive number and integer, respectively, such that $\epsilon \leq 1$, $t\ll 1/\epsilon$, and $t\log t \geq \log (1/\epsilon)$. Additionally, let $k$ be an integer grater than $3\log (t^2/\epsilon)$. 
    We define $U_p$ with $p\sim\mathcal{P}$ is a random unitary operator that maps $\rho^{(t)}_{\mathrm{unique},k}$ to $\rho^{(t)}_{\mathrm{unique},n}$ with a failure probability smaller than $\epsilon$. Then, $U_p$ can be implemented in $O(t[\log t]^3  \log n \log(1/\epsilon))$ depth without any ancilla bits and $\tilde{O}(t [\log t]^2 \log n \log(1/\epsilon))$ depth with $\lceil n/k\rceil$ ancilla bits, where $\tilde{O}(\cdot)$ neglects $\log\log t$ terms.
\end{lemma}
\begin{proof}
    We first note that sampling a type vector $\mathbf{T}$ from $\mathcal{U}_{k,t}$ uniformly at random is equivalent to sampling $t$ distinct uniformly random copies of $\{0,1\}^k\times 0^{n-k}$ and constructing the type vector using the sampled copies. This holds for all $1\leq k \leq n$. Second, $\rho^{(t)}_{\mathrm{unique},n}$ can be interpreted as the ensemble of unique type states in the $t$-copy of the entire $n$-qubit space with the uniform distribution. These imply that $U_p^{\otimes t}$ maps a truly random unique type state $\ket{\mathbf{T}_k}$ in the $t$-copy of a $k$-qubit subspace to a truly random unique type state $\ket{\mathbf{T}_n}$ in the $t$-copy of the entire $n$-qubit system, \textit{i.e.}, 
    \begin{equation}
        U_p^{\otimes t}\ketbra{\mathbf{T}_k}\otimes\ketbra{0}^{\otimes (n-k)t} U_p^{\otimes t,\dag} = \ketbra{\mathbf{T}_n}
    \end{equation}
    if $p$ uniformly randomize bits of $t$ distinct uniformly random copies of $\{0,1\}^k\times 0^{n-k}$. We claim that \textbf{Algorithm 1} with $m=\lceil\log_2 t\rceil$, $k\geq m$, and $\alpha=O(t\log (t/\epsilon))$ implements such $U_p$ with a failure probability smaller than $\epsilon$. 
    
    \textbf{Algorithm 1} is comprised of two main iteration processes. We consider the case when the input state is $\ket{\mathbf{T}_k}\otimes\ket{0}^{\otimes(n-k)t}$ with a unique type state $\ket{\mathbf{T}_k}$ in the $t$-copy of the $k$-qubit subspace. We assume that $n$ is $2^l k$ for some positive integer $l$. 
    
    The first iteration copies the first $k$-qubits, the qubits where $\ket{\mathbf{T}_k}$ lives in, to $n/2k-1$ distinct $k$-qubit registers in the other qubits using Controlled NOT (CNOT) gates as illustrated in \hyperlink{fig:copy_circuit}{Supplementary Figure 1} up to permutations of registers. This iteration can be done in $O(\log (n/k))$ depth.
    
    The second iteration randomly applies $m$-MCX gates $\alpha$ times, each of which is conditioned on bits in each register overwritten by the first iteration, targeting a bit in the other registers, as illustrated in \hyperlink{fig:randomize_circuit}{Supplementary Figure 2} up to permutations of registers. These random applications of MCX gates randomize targeted bits. We will prove this later. The second iteration can be done in $O(\alpha \beta k \log (n/k))$ depth, where $\beta$ is the number of gates required to implement each MCX gate. The final process of the algorithm is to update the first register using MCX gates conditioned on the second register, which can be done in $O(\alpha \beta k)$ depth. Thus, the total depth is given by $O(\alpha\beta k \log(n/k))$.

    Let us prove that random applications of MCX gates indeed randomize targeted bits. To this end, let us assume that we have sampled a set of $m$-distinct bits of a register updated by the first iteration and random $m$-bits $\alpha$ times using \textbf{Algorithm 2}. As we defined in the previous paragraph, the sequences of these sets and random bit-strings are denoted by $S$ and $C$, respectively. Let us apply $m$-MCX gates based on $S$ and $C$. Specifically, the $i$-th MCX gate has $S[i]$ as control bits and $C[i]$ as control conditions. We then choose a bit that has not been updated as the target bit. Then, we apply each of the MCX gates with the probability $1/2$. 
    
    Now, let us assume that the input state $\{x_i\}_{i=1}^t$ is sampled from the set of $t$ distinct copies of $\{0,1\}^k\times 0^{n-k}$ uniformly at random. If the conditions of an MCX gate match with bits of a copy $x_i$, then the random application of the MCX gate assigns a binary random variable to the target bit value of the copy, as illustrated in \hyperlink{fig:MCX-random-var}{Supplementary Figure 3}. By applying the MCX gates randomly, the target bit values of all copies may acquire mutually independent binary random variables. If this happens, then the target bit is randomized up to $t$-th moments. Let us check if this is indeed the case.

    Let $x_i|S[j]$ be bits of $x_i$ in $S[j]$ such that $(x_i|S[j])_l=(x_i)_{(S[j])_l}$ for all $l\in[1,m]$. Let $v_j$ be a vector in $\mathbb{F}^t_2$ such that $(v_j)_i=\delta_{x_i|S[j],C[j]}$. This means that $(v_j)_i$ is nontrivial when the conditions of the $i$-th MCX gate match with the bits of the $j$-th copy. Let us construct the $t\times \alpha$ matrix $X$ over $\mathbb{F}_2$ whose columns are $\{v_j\}_{j=1}^\alpha$. If $X$ is full-ranked, then the target bit values of the copies acquire copy-wise independent binary random variables. In \hyperlink{sec:full-rank-prob}{Supplementary Notes 3}, we show that for $k\geq 3\log(t^2/\epsilon)$, $X$ is full-ranked with a failure probability smaller than $\epsilon$ given that $\alpha=t\log(1/\epsilon)$.

    Since we use the same set of MCX gates for randomizing all registers conditioned on the copied registers, all registers except the first one are randomized with a failure probability smaller than $\epsilon$. The first register is updated by MCX gates conditioned on $t$-wist independent random $k$-bits. As discussed in \hyperlink{sec:full-rank-prob}{Supplementary Notes 3}, such MCX gates also randomize the first register if $X$ is full-ranked. Thus, all registers are randomized with a failure probability smaller than $\epsilon$.

    Now, let us compute the depth upper bounds of the circuit generated by \textbf{Algorithm 1}. Since each MCX gate with $m$ control bits can be implemented in $O(m)$ depth without ancilla bits~\cite{gidney2015} and $\Theta((\log m)^3)$ depth with a single ancilla~\cite{Claudon2024}, the total depth of a circuit implementing \textbf{Algorithm 1} is given by $O\left(t \log t \log(t^2/\epsilon) \log(1/\epsilon)\log(n/k) \right)$ without ancilla bits and $\tilde{O}\left(t \log(t^2/\epsilon)\log(1/\epsilon) \log(n/k) \right)$ with $\lceil n/k \rceil$ ancilla bits.

    We note that MCX gates in $\{U_p\}$ can be applied in parallel. When $t\ll 1/\epsilon$ and $t\log t\gg \log (1/\epsilon)$, this reduces the depth required to implement $U_p$ by factor $[\log t]^2/\log(1/\epsilon)$ with a failure probability smaller than $\epsilon$. More details on this MCX parallelization can be found in \hyperlink{sec:parallel-MCX}{Supplementary Notes 4}. This consequently gives depth upper bounds claimed in the theorem.
\end{proof}

\begin{lemma}\label{thm:perm-avg-to-n-unique-state}
    Let $U_{p_a}$
    with $p_a\sim\mathcal{P}_a$ be a unitary operator that implements a $t$-wise independent random injective map $[2^k]\rightarrow[2^n]$ with a failure probability smaller than $\epsilon/2$. Then, it follows that
    \begin{equation*}
        \operatorname{TD}\left(\mathbb{E}_{p_a\sim\mathcal{P}_a}\left[\mathcal{C}^{\otimes t}_{p_a}\left(\rho^{(t)}_{\mathrm{unique},k}\otimes\ketbra{0}^{\otimes(n-k)t}\right)\right],\rho^{(t)}_{\mathrm{unique},n}\right)\leq \epsilon
    \end{equation*}
    with $\mathcal{C}_{p_a}(\rho) = U_{p_a}\rho U_{p_a}^\dagger$.
\end{lemma}
\begin{proof}
    We first note that the following is true for a $t$-wist independent random injective map $p\sim\mathcal{P}$ due to \textbf{Lemma}~\ref{thm:sub-global-unique}: 
    \begin{equation}
        \mathbb{E}_{p\sim\mathcal{P}} \left[\mathcal{C}^{\otimes t}_p\left(\rho^{(t)}_{\mathrm{unique},k}\otimes\ketbra{0}^{\otimes(n-k)t}\right)\right] = \rho^{(t)}_{\mathrm{unique},n}.
    \end{equation}
    Since $U_{p_a}$ implements the $t$-wise independent random injective map with the failure probability $p_\mathrm{fail}\leq\epsilon$, we have
    \begin{equation}
        \mathbb{E}_{p_a\sim\mathcal{P}_a} \left[\mathcal{C}^{\otimes t}_{p_a}\left(\rho^{(t)}_{\mathrm{unique},k}\otimes\ketbra{0}^{\otimes(n-k)t}\right)\right] = (1-p_\mathrm{fail})\rho^{(t)}_{\mathrm{unique},n} + p_\mathrm{fail} \sigma_\mathrm{E},
    \end{equation}
    where $\sigma_\mathrm{E}$ is an erroneous state. Therefore, it follows that
    \begin{equation}
        \begin{split}
            \operatorname{TD}\left(\mathbb{E}_{p_a\sim\mathcal{P}_a}\left[\mathcal{C}^{\otimes t}_{p_a}\left(\rho^{(t)}_{\mathrm{unique},k}\otimes\ketbra{0}^{\otimes(n-k)t}\right)\right],\rho^{(t)}_{\mathrm{unique},n}\right)
            &=\mathrm{TD}\left((1-p_\mathrm{fail})\rho^{(t)}_{\mathrm{unique},n} + p_\mathrm{fail} \sigma_\mathrm{E}, \rho^{(t)}_{\mathrm{unique},n}\right)\\
            &\leq 2p_\mathrm{fail}\\
            &\leq \epsilon.
        \end{split}
    \end{equation}
\end{proof}

Finally, let us restate \textbf{Theorem 3} and prove it.
\begin{theorem}[\textbf{Theorem 3} {in the main text}]\label{thm:parallel-approximate-design-depth}
    Let $\epsilon$ and $t$ be a positive number and integer, respectively, such that $\frac{t^2}{2^{n+3}}\leq \epsilon \leq 1/2$, $t\ll 1/\epsilon$, and $t\log t \geq \log (1/\epsilon)$. Then, $\epsilon$-approximate state $t$-design having $O(\log(t/\epsilon))$ entanglement, magic, and coherence can be implemented in $O(t[\log t]^3\log n\log(1/\epsilon))$ depth without ancilla bits and $\tilde{O}(t [\log t]^2 \log n \log(1/\epsilon))$ depth with $\lceil n/\lceil 3\log_2(t^2/\epsilon)\rceil\rceil$ ancilla bits where $\tilde{O}(\cdot)$ neglects $\log\log t$ terms.
\end{theorem}
\begin{proof}
    Let us set $k=\lceil 3\log_2(t^2/\epsilon)\rceil$. An ensemble forming $(\epsilon/8)$-approximate state $t$-design in a $k$-qubit subsystem can be implemented by $O(t[\log t]^7\log k \log (1/\epsilon))$ depth circuits~\cite{schuster2024}. We will neglect this circuit depth as it is much smaller than that required to implement the random injective map. Due to \textbf{Lemma}~\ref{thm:Haar-unique}, the $t$-th moments $\rho^{(t)}_{\mathrm{sub},k}$ of the states in this ensemble approximates $\rho^{(t)}_{\mathrm{unique},k}$ as
    \begin{equation}
        \mathrm{TD}\left(\rho^{(t)}_{\mathrm{sub},k},\rho^{(t)}_{\mathrm{unique},k}\right)\leq \frac{\epsilon}{8} + \frac{t^2}{2^{k-1}} \leq \frac{\epsilon}{4}.
    \end{equation}
    Let $U_{p_a}$ with $p_a\sim\mathcal{P}_a$ be a unitary operator implementation of a $t$-wise independent random injective map that maps $\rho^{(t)}_{\mathrm{unique},k}$ to $\rho^{(t)}_{\mathrm{unique},n}$ with a failure probability smaller than $\epsilon/4$. Due to \textbf{Lemma}~\ref{thm:implement-U_p}, we can implement $U_{p_a}$ using $O(t[\log t]^3 \log n \log(1/\epsilon))$ depth circuits without any ancilla bits and $\tilde{O}(t[\log t]^3\log n\log(1/\epsilon))$ depth circuits with $\lceil n/k\rceil$ ancilla bits. This, together with \textbf{Lemma}~\ref{thm:perm-avg-to-n-unique-state}, \textbf{Lemma}~\ref{thm:Haar-unique}, the triangle inequality, and the strong convexity of the trace distance, gives
    \begin{equation}
        \begin{split}
            &\mathrm{TD}\left(\mathbb{E}_{p_a\sim\mathcal{P}_a}\left[ U_{p_a}^{\otimes t} \rho^{(t)}_{\mathrm{sub},k}\otimes\ketbra{0}^{\otimes(n-k)t} U_{p_a}^{\otimes t,\dag} \right], \rho^{(t)}_{\mathrm{Haar},n}\right)\\
            &\leq \mathrm{TD}\left(\mathbb{E}_{p_a\sim\mathcal{P}_a}\left[ U_{p_a}^{\otimes t} \rho^{(t)}_{\mathrm{sub},k}\otimes\ketbra{0}^{\otimes(n-k)t} U_{p_a}^{\otimes t,\dag} \right], \rho^{(t)}_{\mathrm{unique},n}\right)+\mathrm{TD}\left(\rho^{(t)}_{\mathrm{unique},n},\rho^{(t)}_{\mathrm{Haar},n}\right)\\
            &\leq \mathrm{TD}\left(\mathbb{E}_{p_a\sim\mathcal{P}_a}\left[ U_{p_a}^{\otimes t} \rho^{(t)}_{\mathrm{sub},k}\otimes\ketbra{0}^{\otimes(n-k)t} U_{p_a}^{\otimes t,\dag} \right], \mathbb{E}_{p_a\sim\mathcal{P}_a}\left[U_{p_a}^{\otimes t}\rho^{(t)}_{\mathrm{unique},k}\otimes\ketbra{0}^{\otimes(n-k)t}U_{p_a}^{\otimes t,\dag}\right]\right)\\
            &\quad+\mathrm{TD}\left(\mathbb{E}_{p_a\sim\mathcal{P}_a}\left[U_{p_a}^{\otimes t}\rho^{(t)}_{\mathrm{unique},k}\otimes\ketbra{0}^{\otimes(n-k)t}U_{p_a}^{\otimes t,\dag}\right],\rho^{(t)}_{\mathrm{unique},n}\right) + \frac{t^2}{2^{n-1}}\\
            &\leq \mathrm{TD}\left(\rho^{(t)}_{\mathrm{sub},k}, \rho^{(t)}_{\mathrm{unique},k}\right)+\frac{\epsilon}{2}+\frac{t^2}{2^{n-1}}\\
            &\leq \epsilon.
        \end{split}
    \end{equation}
\end{proof}

\subsubsection*{Efficient circuit implementation of random injective map for $t\leq 3$}
The circuit construction for a random injective map presented in \textbf{Lemma}~\ref{thm:implement-U_p} uses MCX gates with $\lceil\log_2 t\rceil$ controls. However, we find that for $t\leq 3$, it is also possible to implement $U_p$ using CNOT gates, not CCX gates. Here, we present this improved implementation.

We first remind that the role of the random injective map $U_{p_a}$ with $p_a\sim\mathcal{P}_a$ in \textbf{Lemma}~\ref{thm:implement-U_p} is to map $\rho^{(t)}_{\mathrm{unique},k}$ to $\rho^{(t)}_{\mathrm{unique},n}$. This can be achieved by the following property of $U_{p_a}$:
\begin{equation}
    \mathbb{E}_{p\sim\mathcal{P}}\left[U_p^{\otimes t}\left(\ketbra{\textbf{T}_k}\otimes \ketbra{0}^{\otimes (n-k)t}\right)U_p^{\otimes t,\dag}\right] = \rho^{(t)}_{\mathrm{unique},n}
\end{equation}
for any $k$-qubit $t$-copy unique type state $\ket{\textbf{T}_k}$. The main idea for implementing such a random injective map is to construct a random circuit that assigns copy-wise independent binary random variables to all qubits of $\ketbra{\textbf{T}_k}\otimes \ketbra{0}^{\otimes (n-k)t}$. 

Our improved circuit has the same structure as our previous circuit. It first divides the entire $n$-qubit system into $2^\beta$ distinct $k$-qubit registers with $\beta=\log_2(n/k)$. Next, it copies the first register, which stores the unique type state $\ketbra{\textbf{T}_k}$ in the $t$-copy space, to the odd registers in $\beta-1$ depth by applying CNOT gates in parallel. It then randomizes qubits in the even registers using random CNOT gates conditioned on qubits in the odd registers. Here, what we mean by ``random'' is {that the decision of whether to apply the CNOT gates is made randomly.}
To be more precise, for all $b\in[k]$ and $l\in[2^{\beta-1}]$, the $b$-th qubit in the $2l$-th register is updated by CNOT gates 
\begin{equation*}
    \{\mathrm{CNOT}_{a+k(2l-2),b+k(2l-1)},X_{a+k(2l-2)}\mathrm{CNOT}_{a+k(2l-2),b+k(2l-1)}X_{a+k(2l-2)}\}_{a=1}^k,
\end{equation*}
whose control qubits are in the $(2l-1)$-th register. These CNOT gates are applied with half probability. Locations of the gates for $k=3$ are illustrated in \hyperlink{fig:improved-circuit}{Supplementary Figure 4}(b). Below, we show that these applications of CNOT gates uniformly randomize target qubits. 

Let $\{x_i\}_{i=1}^3$ be distinct bitstrings of length $k$ and $\{y_i\}_{i=1}^3$ be any bits. We consider the following initial state
\begin{equation}
    \ket{\psi} = \bigotimes_{i=1}^3\ket{x_i,y_i}.
\end{equation}
We want to assign copy-wise independent binary random variables on the $(k+1)$-th qubit of $\ket{\psi}$. To this end, we apply the following unitary operator implementing the random CNOT gate circuit illustrated in \hyperlink{fig:improved-circuit}{Supplementary Figure 4}(b) to $\ket{\psi}$: 
\begin{equation}
    U = \prod_{j=1}^k \mathrm{CNOT}_{j,k+1}^{\hat{X}_j} \left(X_j \mathrm{CNOT}_{j,k+1} X_j\right)^{\hat{X}'_j},
\end{equation}
where $\hat{X}_i$ and $\hat{X}'_i$ are independent uniform binary random variables. Then, $\ket{\psi}$ becomes
\begin{equation}
    U^{\otimes 3}\ket{\psi} = \bigotimes_{i=1}^{3} \ket{x_i,\hat{F}_i}
\end{equation}
with 
\begin{equation}
    \hat{F}_i \equiv_2 \sum_{j=1}^k [\delta_{(x_i)_j,1}\hat{X}_j+\delta_{(x_i)_j,0}\hat{X}'_j]+y_i.
\end{equation}
Here, $\equiv_2$ represents the modulo two equivalence. We note that $\{\hat{F}_i\}_{i=1}^3$ are linearly independent since $\{x_i\}_{i=1}^3$ are distinct. To show this, let us try to make them linearly dependent and check that it is impossible. Since $\{x_i\}_{i=1}^3$ are distinct, there exists $j_1\in[1,k]$ such that $(x_1)_{j_1}=0$ and $(x_2)_{j_1}=1$. This makes $\hat{F}_1$ and $\hat{F}_2$ linearly independent as $\hat{X}'_{j_1}$ in $\hat{F}_1$ and $\hat{X}_{j_1}$ in $\hat{F}_2$ cannot canceled each other. The value of $(x_3)_{j_1}$ is either $0$ or $1$. Without loss of generality, let us assume that $(x_3)_{j_1}=0$. This implies that $\hat{F}_2$ and $\hat{F}_3$ are linearly independent. Now, let us try to make $\hat{F}_1$ and $\hat{F}_3$ linearly dependent. Since $x_1$ and $x_3$ are distinct, there exists $j_2\in[1,k]$ such that $j_2\neq j_1$ and $(x_1)_{j_2}\neq(x_3)_{j_2}$. This implies that $\hat{F}_1+\hat{F}_3$ has $\hat{X}_{j_2}$ as well as $\hat{X}'_{j_2}$. These two variables, however, cannot be canceled by $\hat{F}_2$ no matter what the value of $(x_2)_{j_2}$ is. Thus, it is impossible to make $\hat{F}_1$ and $\hat{F}_3$ linearly dependent. Thus, $\{\hat{F}_i\}_{i=1}^3$ are independent. In addition to the independence, it is worth to note that any sum of independent uniform binary random variables is a uniform binary random variable. Thus, $\{\hat{F}_i\}_{i=1}^3$ are independent uniform random binary variables, \textit{i.e.}, 
\begin{equation}
    \mathrm{Prob}\left(\hat{F}_1=z_1,\hat{F}_2=z_2,\hat{F}_3=z_3\right) = \frac{1}{2^3}
\end{equation}
for all $\{z_i\}_{i=1}^3\in[2]^3$. In this way, the random CNOT circuit randomizes target qubits. If the initial state $\ket{\psi}$ has more qubits that need to be randomized, then we can also randomize those qubits by repeating this process. 

After this stage, the unique type state in the $k$-qubit subspace becomes
\begin{equation}
    \begin{split}
        &\ketbra{\textbf{T}_k}\otimes\ketbra{0}^{\otimes(n-k)t} = \frac{1}{t!} \sum_{\sigma,\pi\in S_t} \ketbra{\sigma((\mathbf{T}_k)_1,\cdots,(\mathbf{T}_k)_t)}{\pi((\mathbf{T}_k)_1,\cdots,(\mathbf{T}_k)_t)} \\
        &= \sum_{\sigma\in S_t} \ketbra{\sigma((\mathbf{T}_k)_1,\cdots,(\mathbf{T}_k)_t)}{(\mathbf{T}_k)_1,\cdots,(\mathbf{T}_k)_t} \\
        &\rightarrow \frac{1}{2^{(n-k)t}}\sum_{\sigma\in S_t}\sum_{\{z_i\}_{i=1}^t\in[2^{n-k}]^t}\ketbra{\sigma((\mathbf{T}_k)_1,\cdots,(\mathbf{T}_k)_t)}{(\mathbf{T}_k)_1,\cdots,(\mathbf{T}_k)_t}\otimes\ketbra{\sigma(z_1,\cdots,z_t)}{z_1,\cdots,z_t}.
    \end{split}
\end{equation}
Finally, we randomize the first register by randomly applying CNOT gates conditioned on the second register and targeting the first register. Since the second register has copy-wise distinct configurations with probability higher than $1-t^2/2^{k+1}$, this uniformly randomizes the first register as
\begin{equation}
    \ketbra{\textbf{T}_k}\otimes\ketbra{0}^{\otimes(n-k)t}\rightarrow \frac{1}{2^{nt}}\sum_{\sigma\in S_t}\sum_{\{z_i\}_{i=1}^t\in[2^n]^t}\ketbra{\sigma(z_1,\cdots, z_t)}{z_1,\cdots,z_t} + O(t^2/2^k)\times \rho_\mathrm{fail},
\end{equation}
where $\rho_\mathrm{fail}$ is the density matrix for failing cases. We note that the failure probability is upper bounded by $t^2/2^k\leq O(\epsilon)$. The first term on the right-hand side approximates the equal summation of all possible unique type states in the entire space, as
\begin{equation}
    \begin{split}
        \frac{1}{2^{nt}}\sum_{\sigma\in S_t}\sum_{\{z_i\}_{i=1}^t\in[2^n]^t}\ketbra{\sigma(z_1,\cdots, z_t)}{z_1,\cdots,z_t} 
        &= \binom{2^n}{t}^{-1}\sum_{\sigma\in S_t}\sum_{\{z_i\}_{i=1}^t\in[2^n]^t_\mathrm{dist}}\ketbra{\sigma(z_1,\cdots, z_t)}{z_1,\cdots,z_t}\\
        &\qquad\qquad\qquad\qquad\qquad\qquad+O(t^2/2^n)\times \rho_\mathrm{err}\\
        &= \rho^{(t)}_{\mathrm{unique},n}+O(t^2/2^n)\times \rho_\mathrm{err}
    \end{split}
\end{equation}
with another density matrix $\rho_\mathrm{err}$ corresponding to the approximation error due to copy-wise collisions between bit strings $\{z_i\}_{i=1}^t$. Thus, by taking $k=\Omega(\log(t^2/\epsilon))$, we can make the approximation error be upper bounded by $O(\epsilon)$.

\newpage
\hypertarget{sec:full-rank-prob}{}
\subsection*{3. Full rank probability}\label{sec:full-rank-prob}
Let us assume that we have sampled $\{x_i\}_{i=1}^t\in[K]^t_\mathrm{dist}$ uniformly at random with $K=2^k$. Here, $[K]^t_\mathrm{dist}$ is the set of $t$ distinct integers in $[1,K]$. Let $\alpha$ be a positive integer. We then sample sets $\{S_i\}_{i=1}^\alpha$ of $m$-bits and their values $\{y_i\}_{i=1}^\alpha\in[\mathbb{Z}_2^m]^\alpha$ randomly. Let $x_i|S_j$ be bits of $x_i$ in $S_j$ such that $(x_i|S_j)_l=(x_i)_{(S_j)_l}$ for all $l\in[1,t]$. Let $v_j$ be a vector in $\mathbb{F}_2^t$ such that $(v_j)_i=\delta_{x_i|S_j,y_j}$. We note three important properties of $(v_j)_i$. First, $(v_j)_i$ is one with the marginal probability $1/2^m$. Second, $(v_j)_i$ and $(v_{j'})_i$ with $j\neq j'$ are independent. Third, $(v_j)_i$ and $(v_j)_{i'}$ with $i\neq i'$ are dependent. 

To see the statistical dependence between $(v_j)_i$ and $(v_j)_{i'}$, let us compute the probability of having $(v_j)_i=(v_j)_{i'}=1$. To achieve this, we first need bits of $x_i$ in $S_j$ and $y_j$ are the same. This happens with the probability $1/2^m$. Next, bits of $x_i$ and $x_{i'}$ should be the same in $S_j$. Let the hamming distance between $x_i$ and $x_{i'}$ be $d(x_i,x_{i'})$. Then, it follows that
\begin{equation}
    \sum_{S_j}\mathrm{Prob}(x_i|S_j=x_{i'}|S_j)\mathrm{Prob}(S_j) = \binom{k-d(x_i,x_{i'})}{m}/\binom{k}{m}.
\end{equation}
Consequently, probability of having $(v_j)_i=(v_j)_{i'}=1$ is given by
\begin{equation}
    \mathrm{Prob}((v_j)_i=(v_j)_{i'}=1)=\frac{1}{2^m}\binom{k-d(x_i,x_{i'})}{m}/\binom{k}{m},
\end{equation}
which differs from $[\mathrm{Prob}((v_j)_i=1)]^2=2^{-2m}$. Thus, $(v_j)_i$ and $(v_j)_{i'}$ with $i\neq i'$ are statistically dependent.

Let us analyze the probability that the vectors $\{v_j\}_{j=1}^\alpha$ contain a subset of $t$ linearly independent vectors in a $t$-dimensional space. This corresponds to determining the probability that the $t \times \alpha$ matrix $X$ (with columns ${v_j}$) achieves full rank.

A key challenge arises from within-column dependencies: for each vector $v_j$, its entries $(v_j)_i$ and $(v_j)_{i'}$ with $i \neq i'$ are statistically dependent. These dependencies propagate to the matrix elements of $X$, rendering traditional full-rank probability analyses, which often assume independent entries, insufficient.

To enable tractable analysis, we define a structured approximation matrix $X_2$ where each column contains at most two nonzero elements. The probability $p_2$ of a column having exactly two nonzero elements will be set as the minimum value of $\mathrm{Prob}((v_j)_i=(v_j)_{i'}=1)$ while respecting typical pairwise distances $d(x_i, x_{i'})$. Additionally, the probability $p_1$ of each element being nonzero is set to ensure $\mathrm{Prob}\left((v_j)_i = 1\right) \leq 1/t$. The full rank probability of $X_2$ is smaller than that of $X$ for two reasons. First, $X_2$ neglects the possibilities of each column having more than two non-trivial elements, which increases the full rank probability as columns are statistically independent. Second, $\mathrm{Prob}\left((v_j)_i = 1\right)$ is upper bounded by $1/t$, making the full rank probability even more smaller. Below, we lower bound the full rank probability of $X_2$, which consequently lower bounds that of $X$.

Let us first set $p_2$. Although $\{x_i\}_{i=1}^t$ are fixed before sampling the sets $\{S_i\}_{i=1}^\alpha$ and $\{y_i\}_{i=1}^\alpha$, these $x_i$ themselves originate from a random sampling process. In this regard, we derive a probabilistic bound for the pairwise distances $d(x_i, x_{i'})$, establishing a typical range within which these distances concentrate. For the following analysis, let us hypothetically consider $\{x_i\}_{i=1}^t$ without the distinct condition. Then, the probability of having $d(x_i, x_{i'})\notin [k/2-\delta,k/2+\delta]$ for some $0<\delta<k/2$ is given due to the Chernoff bound by
\begin{equation}
    \mathrm{Prob}\left(\abs{d(x_i,x_{i'})-k/2}>\delta\right) \leq 2e^{-2\delta^2/k}.
\end{equation}
Let $\mathcal{A}$ be a set of events that $\{x_i\}_{i=1}^t$ are typical, \textit{i.e.}, $\forall i\neq i'\in[t]$, $\abs{d(x_i,x_{i'})-k/2}\leq \delta$. Then, due to the union bound, the probability of any event in $\mathcal{A}$ not happening is upper bounded by
\begin{equation}
    \mathrm{Prob}(\mathcal{A}^c) \leq 2\binom{t}{2}e^{-2\delta^2/k}.
\end{equation}
Next, let $\mathcal{B}$ be a set of events that $\{x_i\}_{i=1}^t$ are all mutually distinct. The probability of an event in $\mathcal{B}$ happening is given by
\begin{equation}
    \mathrm{Prob}(\mathcal{B}) = \frac{1}{K^t}\frac{K!}{(K-t)!}.
\end{equation}
Now, we are ready to compute the conditional probability of
\begin{equation}
    \mathrm{Prob}(\mathcal{A}|\mathcal{B}) = \frac{\mathrm{Prob}(\mathcal{A}\cap\mathcal{B})}{\mathrm{Prob}(\mathcal{B})},
\end{equation}
which is our original interest with the distinct condition on $\{x_i\}_{i=1}^t$. Let us consider an event in $\mathcal{B}^c$. For such an event, there exist $i\neq i'\in[t]$ such that $x_i=x_{i'}$. These bitstrings have $\abs{d(x_i,x_{i'})-k/2}=k/2$, violating the typicallity condition given that $\delta<k/2$. Therefore, we have $\mathcal{A}\cap\mathcal{B}=\mathcal{A}$. Thus, $\mathrm{Prob}(\mathcal{A}|\mathcal{B})$ is lower bounded by
\begin{equation}
    \mathrm{Prob}(\mathcal{A}|\mathcal{B}) 
    > \frac{K^t (K-t)!}{K!} (1 - t(t-1)e^{-2\delta^2/k})\geq 1 - t^2 e^{-2\delta^2/k}.
\end{equation}
The failure probability of having an event in $\mathcal{A}|\mathcal{B}$ becomes below $\epsilon>0$ when $\delta$ is given by
\begin{equation}\label{eq:delta-condition}
    \delta^2 \geq \log(t/\sqrt{\epsilon}) k.
\end{equation}
Then, for any choice of $i\neq i'\in[t]$ and with the failure probability below $\epsilon$, randomly chosen conditional bits of $x_i$ and $x_{i'}$ are the same with the worst case probability given by 
\begin{equation}
    \binom{k/2-\delta}{m}/\binom{k}{m} 
    =\frac{1}{2^m}\prod_{i=0}^{m-1} \left(\frac{k-2\delta-2i}{k-i}\right)\frac{1}{2^m}\left(1-\frac{2\delta}{k}\right)^m.
\end{equation}
Based on these, we conclude that with a failure probability below $\epsilon$, $p_2 = 1/(at)^2$ with
\begin{equation}\label{eq:a_bound}
    a\geq \left(1-\frac{2\delta}{k}\right)^{-m/2},
\end{equation}
and $m=\log_2 t$. Let us set $k$ by $c\log_2 (t/\sqrt{\epsilon})$ for some $c>0$. Then, due to Eq.~\eqref{eq:delta-condition}, $\delta$ should be lower bounded by
\begin{equation}
    \delta\geq \sqrt{\frac{c}{\log 2}}\log_2(t/\sqrt{\epsilon}).
\end{equation}
Since $\delta$ should be less than $k/2$ by its definition, $c$ should satisfy 
\begin{equation}
    c \geq \frac{4}{\log 2} \approx 5.77.
\end{equation}
Consequently, taking $k = \lceil c\log_2 (t/\sqrt{\epsilon}) \rceil = \lceil (c/2) \log_2 (t^2/ \epsilon) \rceil \geq \lceil 2.885 \log_2 (t^2/ \epsilon) \rceil$ is necessary to guarantee $t$-wise independence of the implemented random injective map.

Next, let us set $p_1$. There are two possibilities of getting $(v_j)_i=1$. One possibility is to sample $(v_j)_i=1$ with $(v_j)_{i'}=0$ for all $i'\neq i$. This happens with the probability $p_1$ by definition. Second possibility is to sample $(v_j)_i=(v_j)_{i'}=1$ for some $i'\neq i$. This happens with the probability $(t-1)p_2$. Thus, the marginal probability $\mathrm{Prob}\left((v_j)_i = 1\right)$ is given by
\begin{equation}
    \mathrm{Prob}\left((v_j)_i = 1\right) = p_1 + (t-1)p_2.
\end{equation}
Since we enforce this less than $1/t$, $p_1$ should satisfy
\begin{equation}
    p_1 \leq \frac{1}{t} - \frac{t-1}{a^2t^2}.
\end{equation}
Thus, $p_1$ can be set by
\begin{equation}
    p_1 = \left(1-\frac{1}{a^2}\right)\frac{1}{t}.
\end{equation}
In summary, we have
\begin{align}
    p_0 &= 1-p_1-p_2\\
    p_1 &= \left(1-\frac{1}{a^2}\right)\frac{1}{t}\\
    p_2 &= \frac{1}{(at)^2},
\end{align}
where $p_0$ is the probability of having zero $v_j$.

Now, let us compute a lower bound of the full rank probability of $X_2$. To this end, we upper bound the probability of $X_2$ having rank less than $t$:
\begin{equation}
    \begin{split}
        \mathrm{Prob}(\mathrm{rank}(X_2)< t) = \mathrm{Prob}(\exists s\neq\mathbf{0} \in \mathbb{F}_2^t, \textrm{ s.t. }\forall j\in[1,\alpha],s^T v_j=0).
    \end{split}
\end{equation}
Due to the union bound, we have
\begin{equation}
    \begin{split}
        \mathrm{Prob}(\mathrm{rank}(X_2)< t) \leq \sum_{s\neq\mathbf{0} \in \mathbb{F}_2^t}\mathrm{Prob}(\forall j\in[1,\alpha],s^T v_j=0).
    \end{split}
\end{equation}
Since columns are statistically independent, it follows that
\begin{equation}\label{eq:rank-union-bound-sum}
    \mathrm{Prob}(\mathrm{rank}(X_2)< t)\leq \sum_{s\neq\mathbf{0} \in \mathbb{F}_2^t}[\mathrm{Prob}(s^T v=0)]^\alpha,
\end{equation}
where we drop the subscript of $v_j$ labeling the column dependency. We note that $\mathrm{Prob}(s^T v=0)$ only depends on the hamming weight of $s$. Thus, we introduce $p(w)$ which is the sum of $\mathrm{Prob}(s^T v=0)$ over $s\in\mathbb{F}_2^t$ with the fixed hamming weight $\|s\|=w$,
\begin{equation}
    p(w)=\mathrm{Prob}(s\in\mathbb{F}^t_2,\|s\|=w,s^T v=0).
\end{equation}
Then, the probability of rank deficiency is upper-bounded by 
\begin{equation}
    \mathrm{Prob}(\mathrm{rank}(X_2)< t)\leq \sum_{w=1}^t \binom{t}{w} [p(w)]^\alpha.
\end{equation}
We can compute $p(w)$ explicitly by the following case study. First, when $v$ has trivial elements, then $s^T v$ is vanishing for any $s\in\mathbb{F}_2^t$. Second, when $v$ has exactly one non-trivial element, then $s^T v$ is vanishing if there exists $l\in[t]$ such that $v_l$ is one with $s_l=0$. Since the weight of $s$ is $w$, there exist $t-w$ different $v$ vectors having $s^Tv=0$. Third, when $v$ has exactly two non-trivial elements, then $s^T v$ is vanishing if there exist $l\neq l'\in[t]$ such that either $v_l=v_{l'}=1$ with $s_l=s_{l'}=1$ or $s_l=s_{l'}=0$. There are $\binom{w}{2}$ different $v$ satisfying the former condition and $\binom{t-w}{2}$ different $v$ satisfying the later condition. Therefore, $p(w)$ is given by
\begin{equation}
    \begin{split}
        p(w)
        &= p_0 + p_1(t-w)+ p_2\left(\binom{t-w}{2}+\binom{w}{2}\right)\\
        &= \left[1 - \left(1-\frac{1}{a^2}\right)\frac{1}{t} - \frac{t(t-1)}{2a^2t^2}\right] + \left(1-\frac{1}{a^2}\right)\frac{t-w}{t} + \frac{t^2-(2w+1)t+2w^2}{2a^2t^2}\\
        &= 1 - \frac{w}{t} + \frac{w^2}{a^2t^2}.
    \end{split}
\end{equation}

The direct computation of the summation in Eq.~\eqref{eq:rank-union-bound-sum} is challenging. Thus, we instead try to find its upper bound. To this end, let us define  
\begin{equation}
    f(w) = \binom{t}{w}[p(w)]^\alpha
\end{equation}
and
\begin{equation}
    L(w) = \log f(w) = \log \binom{t}{w} + \alpha \log p(w).
\end{equation}
Additionally, let us define $\delta L(w)$ as
\begin{equation}
    \delta L(w) = L(w+1)-L(w).
\end{equation}
This is the ratio between consecutive summands. Now, let us compute $\delta L(w)$. First, the difference between log binomial coefficients is given by
\begin{equation}
    \log \binom{t}{w+1} - \log \binom{t}{w} = \log \frac{t-w}{w+1}.
\end{equation}
Next, let us compute the difference between $\log p(w)$. First, $p(w+1)$ is given by
\begin{equation}
    p(w+1) = p(w) - \frac{1}{t} + \frac{2w+1}{a^2 t^2}.
\end{equation}
Their log difference is then given by
\begin{equation}
    \log \frac{p(w+1)}{p(w)} = \log \left(1 - \frac{a^2t-2w-1}{a^2t^2-a^2 wt+w^2}\right).
\end{equation}
This is maximized when $w=0$, which gives
\begin{equation}
    \log \frac{p(w+1)}{p(w)} 
        \leq \log \left(1 - \frac{a^2 t-1}{a^2 t^2}\right).
\end{equation}
Combining altogether, $\delta L(w)$ is upper bounded by
\begin{equation}
    \delta L(w) \leq \log \frac{t-w}{w+1} - \alpha\frac{a^2t-1}{2a^2t^2}.
\end{equation}
Let us introduce $b>0$ such that 
\begin{equation}
    \alpha=\frac{2(b+1)a^2 t^2}{a^2t-1}\log t.
\end{equation}
Then, we have
\begin{equation}
    \delta L(w) 
    \leq \log \frac{t-w}{w+1} - (b+1) \log t\leq -b\log t.
\end{equation}
This means that $f(w)$ decreases as
\begin{equation}
    f(w+1) \leq f(w)\cdot t^{-b}.
\end{equation}
Using this, we can bound the summation as
\begin{equation}
    \begin{split}
        \sum_{w=1}^t f(w) 
        &\leq f(1) \sum_{w=1}^t t^{-b(w-1)}\\
        &= f(1) \frac{1-t^{-bt}}{1-t^{-b}}.
    \end{split}
\end{equation}
Here, $f(1)$ is upper bounded by
\begin{equation}
    \begin{split}
        f(1) 
        &= t \left(1-\frac{1}{t}+\frac{1}{a^2t^2}\right)^\alpha\\
        &\leq t e^{-\alpha\left(\frac{1}{t}-\frac{1}{a^2 t^2}\right)}\\
        &=t^{1-2(b+1)\left(1-\frac{1}{a^2t}\right)}\\
        &\approx t^{-1-2b}
    \end{split}
\end{equation}
for $1/\epsilon \gg 1$, which is an assumption taken in \textbf{Lemma}~\ref{thm:implement-U_p}. Thus, we finally have
\begin{equation}
    \begin{split}
        \mathrm{Prob}(\mathrm{rank}(X)< t)
        &\lesssim t^{-1-2b}\frac{1-t^{-bt}}{1-t^{-b}}\\
        &\leq t^{-2b}.
    \end{split}
\end{equation}
For the error bound $\epsilon$, we need $b$ to be
\begin{equation}
    b \geq \frac{\log (1/\epsilon)}{2\log t}.
\end{equation}
Then, the number of columns can be set by
\begin{equation}
    \alpha \gtrsim t\log(1/\epsilon).
\end{equation}

\newpage
\hypertarget{sec:parallel-MCX}{}
\subsection*{4. Parallel execution of MCX gates}\label{sec:parallel-MCX}
Let us set $t>0$, $1\geq\epsilon>0$, and $\beta=O(1)$. Additionally, let us set $k=3\log (t^2/\epsilon)$, $m=\log t$, and $\alpha=\beta t\log(1/\epsilon)$. Let us consider an asymptotic limit of $1/\epsilon\gg t$. Then, we have $k\sim \log(1/\epsilon)$ and $\alpha\sim \beta t\log(1/\epsilon)$. Let us assume that there are two $k$-bit registers $\mathrm{A}$ and $\mathrm{B}$. For each bit in $\mathrm{B}$, let us apply an $m$-MCX gate conditioned on random bits in $\mathrm{A}$ targeting the bit in $\mathrm{B}$ and repeat this $\alpha$ times. We note that the same set of conditions is used to apply MCX gates on different bits in $\mathrm{B}$. Since these MCX gates commute with each other, we can apply MCX gates with distinct conditional and target bits in parallel, thereby reducing the depth required to apply all the MCX gates.

The minimum depth achievable by the parallelization of commuting gates is deeply related to the edge coloring problem of a graph~\cite{Bremner2017}. More precisely, let us consider an $m$-uniform hypergraph $H$ with $k$ vertices. Hyperedges of $H$ can be interpreted as $m$-qubit commuting gates. Then, the chromatic index $\chi'(H)$ of $H$ is exactly the minimum depth of applying all the gates. This is because edges having the same color can be executed in parallel. In our case, since MCX gates target bits in $\mathrm{B}$, the number of MCX gates that can be simultaneously applied is limited to $k$. Given that the number of edges having the same color is upper bounded by $k/m$, we can achieve the minimum depth $\chi'(H)$. 

The direct computation of $\chi'(H)$ is hard in general. Instead, there are some works that study its upper bounds. A general upper bound on $\chi'(H)$ would be
\begin{equation}
    \chi'(H) \leq m(\Delta(H)-1)+1.
\end{equation}
This can be derived using a greedy coloring strategy, meaning that when coloring an edge $e$, we choose a color that is different from all previously colored edges that intersect $e$. To complete the upper bound, let us compute $\Delta(H)$. Let us consider a vertex $v$ in $H$. Each edge of $H$ contains $v$ with the probability of $m/k$. Since each edge is sampled independently, the degree $d(v)$ of the vertex $v$ is given as 
\begin{equation}
    d(v) = \sum_{i=1}^\alpha \hat{Y}_i,
\end{equation}
where $\{\hat{Y}_i\}$ are independent binary random variables with the probability $m/k$ of being one. Then, the mean value of $d(v)$ is given by $\alpha m/k$, and the failure probability of $d(v)$ being concentrated near its mean value is given by
\begin{equation}
    \mathrm{Prob}\left(\abs{d(v)-\frac{\alpha m}{k}} \geq \frac{\delta\alpha m}{k}\right) \leq 2 e^{- \frac{\delta^2\alpha m}{3k}}
\end{equation}
with $0<\delta\leq 1$ due to the Chernoff bound. Since $\alpha m/k$ scales with $t\log t$, we have
\begin{equation}
    \Delta(H) \sim \frac{\beta}{3} t \log t
\end{equation}
with a failure probability $O(t^{-\delta^2\beta t/3})$. When $\log(1/\epsilon)\leq \frac{\delta^2\beta}{3} t\log t$, the failure probability is upper bounded by $\epsilon$. In this regime, we finally get
\begin{equation}
    \chi'(H) \lesssim \frac{\beta}{3\log 2} t (\log t)^2.
\end{equation}

A tighter upper bound on $\chi'(H)$ for random uniform hypergraphs with some conditions described below is established by Ref.~\cite{Kurauskas2015} as
\begin{equation}
    \chi'(H) \leq \frac{m\alpha}{k}(1+\varepsilon)
\end{equation}
for any $\varepsilon>0$ with probability at least $1-\frac{2}{k}-\frac{2k}{m\alpha}$. The conditions they used are
\begin{equation}
    m = o\left(\log\left(\frac{m\alpha}{k\log k}\right)\right)
\end{equation}
and
\begin{equation}
    m = o\left(\log\left(\frac{k}{\log(m\alpha/k)}\right)\right).
\end{equation}
In the asymptotic limit we have considered, these become
\begin{equation}
    \log(1/\epsilon) = o(t^{\beta})
\end{equation}
and
\begin{equation}
    t\log(\beta t\log t) = o(\log(1/\epsilon)),
\end{equation}
respectively. One may see that if $\log(1/\epsilon)$ scales with $t^{1+c}$ for some $c>0$, and $\beta$ is set by $\beta > 1+c$, then these two conditions are simultaneously satisfied. Thus, in this regime, we have
\begin{equation}
    \chi'(H) \lesssim t\log t
\end{equation}
with probability at least $1-O(1/\log t)$. We conjecture that in this regime, $\chi'(H)$ is upper bounded by $O(t\log t)$ with a failure probability smaller than $\epsilon$.

\newpage
\subsection*{5. Lower bounds on entanglement, magic, and coherence (proof of \textbf{Theorem 2} in the main text)}
In this section, we derive upper and lower bounds for entanglement, magic, and coherence of states in an ensemble $\mathcal{E}$ forming $\epsilon$-approximate state $t$-design. For a measure of entanglement, we use the entanglement entropy of any subsystem $A$ such that $|A|=\Theta(n)$ with the number of qubits $n$. For a measure of magic, we use the stabilizer R\'enyi entropy~\cite{Leone2022}. For a $n$-qubit pure state $\ket{\psi}$, this is given as
\begin{equation}
    M_\alpha(\ket{\psi}) = S_\alpha(\Xi_\psi)-\log 2^n,
\end{equation}
where $S_\alpha(\Xi_\psi)$ is the $\alpha$-R\'enyi entropy of the probability distribution $\Xi_\psi$ defined as
\begin{equation}
    \Xi_\psi = \{2^{-n}\bra{\psi}P\ket{\psi}^2|P\in\mathbb{P}_n\}
\end{equation}
with the Pauli group $\mathbb{P}_n$ on $n$ qubits. For a measure of coherence, we use the relative entropy of coherence~\cite{Baumgratz2014}, which is defined as
\begin{equation}
    C_{\mathrm{rel. ent.}}(\rho) = S(\rho_\mathrm{diag}) - S(\rho),
\end{equation}
where $\rho_\mathrm{diag}$ is the matrix made by eliminating all off-diagonal elements of a density matrix $\rho$, and $S(\cdot)$ is the von Neumann entropy. For a pure state $\ket{\psi}$, this becomes $C_{\mathrm{rel. ent.}}(\ketbra{\psi}) = S(\ketbra{\psi}_\mathrm{diag})$. 

Let us first derive lower bound of the entanglement entropy by following Ref.~\cite{aaronson2023quantum}. We first make $t/2$ pairs of states. We then perform $t/2$ SWAP test on the subsystem $A$. Each test successes with probability $1/2+\delta$, where
\begin{equation}
    \delta = \mathbb{E}_{\psi\sim\mathcal{E}}[\tr(\rho^2_A)]/2
\end{equation}
and $\rho_A=\tr_{A^c}(\ketbra{\psi})$, if the underlying ensemble is $\mathcal{E}$. For Haar random states, the success probability is $1/2$ up to exponentially small correction with high probability. Then, the task becomes to detect a binary random variable with a bias from $t/2$ samples. When the bias $\delta$ is small, we can detect this bias based on the direct estimation of the success probability or the likelihood test with probability $1/2+O(\delta\sqrt{t})$. Due to Jensen's inequality, $\delta$ is lower bounded by
\begin{equation}
    \delta \geq \Omega \left(\frac{1}{2^{\mathbb{E}_{\psi\sim\mathcal{E}}[S(\rho_A)]}}\right).
\end{equation}
Combining altogether, we get
\begin{equation}
    \mathbb{E}_{\psi\sim\mathcal{E}}[S(\rho_A)] \geq \Omega(\log (t/\epsilon)).
\end{equation}

Next, we obtain a lower bound of the stabilizer R\'enyi entropy by following Ref.~\cite{gu2023little}. Let us define $\Pi^{2\alpha}=2^{-n}\sum_{P\in\mathbb{P}_n}P^{\otimes 2\alpha}$. From the definition of $M_\alpha(\ket{\psi})$, we have
\begin{equation}
    \tr(\Pi^{(2\alpha)}\ketbra{\psi}^{\otimes 2\alpha}) = 2^{(1-\alpha)M_\alpha(\ket{\psi})}.
\end{equation}
For an odd $\alpha$, a Hadamard test with $\Pi^{(2\alpha)}$ has the success probability of $1/2+\tr(\Pi^{(2\alpha)}\ketbra{\psi}^{\otimes 2\alpha})/2$~\cite{Haug2024}. Using the same logic used for entanglement, for constant $\alpha$, we get
\begin{equation}
    \mathbb{E}_{\psi\sim\mathcal{E}}\left[M_\alpha(\ket{\psi})\right]\geq \Omega\left(\log(t/\epsilon)\right).
\end{equation}


Finally, we can lower bound the relative entropy of coherence by the entanglement entropy. First, every diagonal matrix in the computational basis is a separable state. Thus, by the definition of the relative entropy, the relative entropy of coherence is lower bounded by the relative entropy of entanglement. Second, for pure states, the relative entropy of entanglement is the same as the entanglement entropy. Therefore, we have
\begin{equation}
    \mathbb{E}_{\psi\sim\mathcal{E}}[C_\mathrm{rel.ent.}(\ketbra{\psi})] \geq \Omega(\log(t/\epsilon)).
\end{equation}

\subsection*{6. Upper bound on magic of $\epsilon$-approximate state $t$-designs in \textbf{Theorem 1}}
An upper bound of $\alpha$-stabilizer R\'enyi entropy $M_\alpha(\ket{\psi})$ of $\ket{\psi}$ in $\mathcal{E}_\mathrm{sub}$ with the subset dimension $K=2^{\Theta(\log(t/\epsilon))}$ can be obtained using \textbf{Lemma S4} of Ref.~\cite{gu2023little}. It states that any random subset phase state having the subset dimension $K$ satisfies $M_0(\ket{\psi})\leq 2\log K$. Using the same proof, one can see that the same holds for any state having $2^k$ superpositions in the computational basis. This together with the hierarchy of the stabilizer R\'enyi entropy imply that for $\alpha>0$, $M_\alpha(\ket{\psi})$ is upper bounded by
\begin{equation}
    M_\alpha(\ket{\psi})\leq M_0(\ket{\psi}) \leq O(\log K) = O(\log (t/\epsilon)).
\end{equation}

\newpage
\subsection*{7. $O(1)$-entangled shadow estimator (proof of \textbf{Theorem 4} in the main text)}
In this section, we drive the shadow estimator for measurement unitary operators 
\begin{equation}
    \mathcal{E}=\{V\otimes I^{\otimes (n-k)}U_p|V\in\mathcal{U}(K),p\in\mathcal{P}(N)\}
\end{equation}
with the system size $n$, the subsystem size $k\leq n$, $N=2^n$, $K=2^k$, a unitary 3-design ensemble $\mathcal{U}(K)$ in the subspace, and the ensemble $\mathcal{P}$ of injective maps from $\{0,1\}^k$ to $\{0,1\}^n$ with a uniform distribution. Let $U_p$ with $p\sim\mathcal{P}$ be the unitary operator such that $U_p\ket{b,a}=\ket{p(b)\oplus 0^ka}$ with $k$-bits $k$ and $(n-k)$-bits $a$. In the subsequent subsection, we discuss how to efficiently construct $U_p$ for the shadow tomography.
We find that the shadow estimators that correspond to this measurement setup are:
\begin{equation}
\begin{aligned}
    \hat\rho_{\rm d.} &= \ketbra{z} &&\text{sampled from}~z\sim\bra z \rho \ket z\\
    \hat\rho_{\rm o.d.} &= \frac{(K+1)(N-1)}{K-1}U^\dag\ketbra{z}U&&\text{sampled from}~z\sim\bra{z}U\rho U^\dag\ket{z}.
\end{aligned}
\end{equation}
Let us decompose $\hat{O}$ into the diagonal and off-diagonal components as $\hat{O} = \hat O_{\rm d.} + \hat O_{\rm o.d.}$. We define $\hat{o}_\mathrm{d.}$ and $\hat{o}_\mathrm{o.d.}$ as $\tr(\hat{\rho}_\mathrm{d.}\hat{O}_\mathrm{d.})$ and $\tr(\hat{\rho}_\mathrm{o.d.}\hat{O}_\mathrm{o.d.})$, respectively. We then provide the proof of \textbf{Theorem 4} in the main text by showing the following two theorems:
\begin{theorem}
    $\hat{o}_\mathrm{d.}$ and $\hat{o}_\mathrm{o.d.}$ are unbiased estimators of $\tr \left( \rho \hat{O}_{\rm d.} \right)$ and $\tr \left( \rho \hat{O}_{\rm o.d.} \right)$, \textit{i.e.},
    \begin{equation}
    \begin{aligned}
        \mathbb{E}_{z\sim \bra z \rho \ket z}\left[ \hat{o}_\mathrm{d.} \right] &= \tr \left( \rho \hat{O}_{\rm d.} \right) \\
        \mathbb{E}_{U \sim {\cal E}, z\sim\bra{z}U\rho U^\dag\ket{z}} \left[ \hat{o}_\mathrm{o.d.} \right] &= \tr \left( \rho \hat{O}_{\rm o.d.} \right).
    \end{aligned}
    \end{equation}
\end{theorem}
\begin{proof}
    We skip proving the statement on $\hat{o}_\mathrm{d.}$. We define two twirls for random injective maps and Haar random unitaries:
    \begin{equation}
        \mathcal{C}^{(t)}_\mathcal{P}(X) = \mathbb{E}_{p\sim \mathcal{P}(N)}[U_p^{\otimes t}XU_p^{\otimes t,\dag}]
    \end{equation}
    and
    \begin{equation}
        \mathcal{C}^{(t)}_{\mathrm{H},k}(X) = \mathbb{E}_{V\sim\mathrm{Haar}(K)}[V^{\otimes t}\otimes I^{\otimes(n-k)t}X V^{\otimes t,\dag}\otimes I^{\otimes(n-k)t}]
    \end{equation}
    with a positive integer $t$ and a $N^{t}\times N^{t}$ matrix $X$. The ensemble average of {$\bra{z}U^\dag \hat{O}_\mathrm{o.d.} U\ket{z}$} over $U\sim\mathcal{E}$ and $z\sim\bra{z}U\rho U^\dag\ket{z}$ is given by
    \begin{equation}\label{eq:shadow-mean-estimation}
        \begin{split}
            &\mathbb{E}_{U\sim\mathcal{E},z\sim\bra{z}U\rho U^\dag\ket{z}}\left[\tr({\hat{O}_\mathrm{o.d.}} U^\dag \ketbra{z} U)\right]\\
            &=\sum_{b\in[K]}\sum_{a\in[N/K]}\tr(\rho\otimes {\hat{O}_\mathrm{o.d.}} \left(\mathcal{C}^{(2)}_\mathcal{P}\circ\mathcal{C}^{(2)}_{\mathrm{H},k}\right)\left(\ketbra{b,a}^{\otimes 2}\right))\\
            &=\frac{1}{K+1}\sum_{a\in[N/K]}\sum_{\{b_1,b_2\}\in[K]^2}\sum_{\tau\in S_2}\tr(\rho\otimes {\hat{O}_\mathrm{o.d.}}\mathcal{C}^{(2)}_\mathcal{P}\left(\ketbra{\tau(b_1,b_2)}{b_1,b_2}\otimes\ketbra{a}^{\otimes 2}\right))\\
            &=\frac{K-1}{(K+1)(N-1)}\sum_{\{z_1,z_2\}\in[N]^2_\mathrm{dist}}\bra{z_1}\rho\ket{z_1}\bra{z_2}{\hat{O}_\mathrm{o.d.}}\ket{z_2}+\frac{K-1}{(K+1)(N-1)}\sum_{\{z_1,z_2\}\in[N]^2_\mathrm{dist}}\bra{z_1}\rho\ket{z_2}\bra{z_2}{\hat{O}_\mathrm{o.d.}}\ket{z_1}\\
            &\quad+\frac{2}{K+1}\sum_{z\in[N]} \bra{z}\rho\ket{z}\bra{z}{\hat{O}_\mathrm{o.d.}}\ket{z}\\
            &=\frac{K-1}{(K+1)(N-1)}\tr(\rho {\hat{O}_\mathrm{o.d.}}),
        \end{split}
    \end{equation}
    where $S_t$ is the symmetric group of degree $t$, and $[2^n]^t_\mathrm{dist}$ is the set of $t$-distinct tuples in $[2^n]$. Here, we use the following identity:
    \begin{equation}
        \mathcal{C}^{(2)}_\mathcal{P}(\ketbra{\tau(b_1,b_2)}{b_1,b_2}\otimes\ketbra{a}^{\otimes 2}) = \frac{1}{N(N-1)}\sum_{\{z_1,z_2\}\in[N]^2_\mathrm{dist}}(1-\delta_{b_1,b_2})\ketbra{\tau(z_1,z_2)}{z_1,z_2} + \frac{1}{N}\sum_{z\in[N]}\delta_{b_1,b_2}\ketbra{z}
    \end{equation}
    for any $\{b_1,b_2\}\in[K]^2$, $a\in[N/K]$, and $\tau\in S_2$. 
\end{proof}

Now, let us upper bound the sample complexities of these estimators.
\begin{theorem}\label{thm:sample-complexity}
    The variances of $\hat{o}_\mathrm{d.}$ and $\hat{o}_\mathrm{o.d.}$ are upper bounded by
    \begin{equation}
        \begin{split}
            \mathrm{Var}(\hat{o}_\mathrm{d.})&\leq \tr(\hat{O}^2)\\
            \mathrm{Var}(\hat{o}_\mathrm{o.d.})&\leq O\left(\tr(\hat{O}^2)\right) + O\left(NK^{-1}\sum_{z\in[N]}\bra{z}\rho\ket{z}\bra{z}\hat{O}^2\ket{z})\right)
        \end{split}
    \end{equation}
\end{theorem}
\begin{proof}
The variances are bounded by the second moments of $\hat{o}_\mathrm{d.}$ and $\hat{o}_\mathrm{o.d.}$:
\begin{equation}
\begin{aligned}
    \mathbb{E}_{z\sim \bra z \rho \ket z}\left[ \hat{o}_{\rm d.}^2 \right] &= \sum_z \bra z \rho \ket z \bra{z} \hat O_{\rm d.} \ket{z}^2 \leq \left[\max_{z \in [N]} \bra{z} \hat{O}_{\rm d.} \ket{z}\right]^2 \leq \tr \left( \hat O^2_{\rm d.}\right) \leq \tr \left( \hat{O}^2 \right)\\
    \mathbb{E}_{U \sim {\cal E}, z\sim\bra{z}U\rho U^\dag\ket{z}} \left[ \hat{o}_{\rm o.d.}^2 \right] &= \frac{(K+1)(K-2)}{(K+2)(K-1)}\frac{N-1}{N-2}\left[\tr({\hat{O}_\mathrm{o.d.}}^2) + 2\tr(\rho {\hat{O}_\mathrm{o.d.}}^2) \right]\\
    &\quad+\frac{4(K+1)}{(K+2)(K-1)}\frac{(N-1)(N-K)}{N-2}\sum_{z\in[N]}\bra{z}\rho\ket{z}\bra{z}{\hat{O}_\mathrm{o.d.}}^2\ket{z}\\
    &\leq O\left(\tr(\hat{O}^2)\right) + O\left(NK^{-1}\sum_{z\in[N]}\bra{z}\rho\ket{z}\bra{z}\hat{O}^2\ket{z}\right).
\end{aligned}
\end{equation}
To get the upper bound for the second moments of $\hat{o}_{\rm o.d.}$, we use the following fact:
\begin{equation}
    \begin{split}
        &\mathbb{E}_{U\sim \mathcal{E},z\sim\bra{z}U\rho U^\dag\ket{z}}\left[\tr({\hat{O}_\mathrm{o.d.}} U^\dag \ketbra{z} U)^2\right]\\
        &=\sum_{b\in[K]}\sum_{a\in[N/K]}\tr(\rho\otimes {\hat{O}_\mathrm{o.d.}}\otimes {\hat{O}_\mathrm{o.d.}}\left(\mathcal{C}^{(3)}_\mathcal{P}\circ\mathcal{C}^{(3)}_{\mathrm{H},k}\right)\left(\ketbra{b,a}^{\otimes 3}\right))\\
        &=\frac{K!}{(K+2)!}\sum_{a\in[N/K]}\sum_{\{b_i\}_{i=1}^3\in[K]^3}\sum_{\tau\in S_3}\tr(\rho\otimes {\hat{O}_\mathrm{o.d.}}\otimes {\hat{O}_\mathrm{o.d.}}\mathcal{C}^{(3)}_\mathcal{P}\left(\ketbra{\tau(b_1,b_2,b_3)}{b_1,b_2,b_3}\otimes\ketbra{a}^{\otimes 3}\right))\\
        &=\frac{(K-1)!}{(K+2)!}\frac{K!}{(K-3)!}\frac{(N-3)!}{(N-1)!}\sum_{\{z_i\}_{i=1}^3\in[N]^3_\mathrm{dist}}\sum_{\tau\in S_3}\tr(\rho\otimes {\hat{O}_\mathrm{o.d.}}\otimes {\hat{O}_\mathrm{o.d.}}\ketbra{\tau(z_1,z_2,z_3)}{z_1,z_2,z_3})\\
        &\quad+\frac{(K-1)!}{(K+2)!}\frac{K!}{(K-2)!}\frac{1}{N-1}\sum_{\{z_i\}_{i=1}^2\in[N]^2_\mathrm{dist}}\sum_{\tau\in S_3}\sum_{\sigma\in C_3}\tr(\rho\otimes {\hat{O}_\mathrm{o.d.}}\otimes {\hat{O}_\mathrm{o.d.}}\ketbra{\tau\circ\sigma(z_1,z_1,z_2)}{\sigma(z_1,z_1,z_2)})\\
        &\quad+K\frac{(K-1)!}{(K+2)!}\sum_{z\in[N]}\sum_{\tau\in S_3}\tr(\rho\otimes {\hat{O}_\mathrm{o.d.}}\otimes {\hat{O}_\mathrm{o.d.}}\ketbra{\tau(z,z,z)}{z,z,z})\\
        &=\frac{(K-1)!}{(K+2)!}\frac{K!}{(K-3)!}\frac{(N-3)!}{(N-1)!}\sum_{\{z_i\}_{i=1}^3\in[N]^3}\sum_{\tau\in S_3}\tr(\rho\otimes {\hat{O}_\mathrm{o.d.}}\otimes {\hat{O}_\mathrm{o.d.}}\ketbra{\tau(z_1,z_2,z_3)}{z_1,z_2,z_3})\\
        &\quad+\frac{(K-1)!}{(K+2)!}\frac{K(K-1)(N-K)}{(N-1)(N-2)}\sum_{\{z_i\}_{i=1}^2\in[N]^2}\sum_{\tau\in S_3}\sum_{\sigma\in C_3}\tr(\rho\otimes \hat{O}_\mathrm{o.d.}\otimes \hat{O}_\mathrm{o.d.}\ketbra{\tau\circ\sigma(z_1,z_1,z_2)}{\sigma(z_1,z_1,z_2)})\\
        &\quad+\frac{(K-1)!}{(K+2)!}\frac{K(N-2K)(N-K)}{(N-1)(N-2)}\sum_{z\in [N]}\sum_{\tau\in S_3}\tr(\rho\otimes {\hat{O}_\mathrm{o.d.}}\otimes {\hat{O}_\mathrm{o.d.}}\ketbra{\tau(z,z,z)}{z,z,z})\\
        &=\frac{(K-1)!}{(K+2)!}\frac{K!}{(K-3)!}\frac{(N-3)!}{(N-1)!}\left[\tr({\hat{O}_\mathrm{o.d.}}^2) + 2\tr(\rho {\hat{O}_\mathrm{o.d.}}^2) \right]\\
        &\quad+4\times\frac{(K-1)!}{(K+2)!}\frac{K(K-1)(N-K)}{(N-1)(N-2)}\sum_{z\in[N]}\bra{z}\rho\ket{z}\bra{z}{\hat{O}_\mathrm{o.d.}}^2\ket{z},
    \end{split}
\end{equation}
and the following equality for the injective map twirl:
\begin{equation}
    \begin{split}
        &\mathcal{C}^{(3)}_\mathcal{P}(\ketbra{\tau(b_1,b_2,b_3)}{b_1,b_2,b_3}\otimes\ketbra{a}^{\otimes 3})\\
        &=\frac{(N-3)!}{N!}\sum_{\{z_i\}_{i=1}^3\in[N]^3_\mathrm{dist}}(1-\delta_{b_1,b_2})(1-\delta_{b_2,b_3})(1-\delta_{b_1,b_3})\ketbra{\tau(z_1,z_2,z_3)}{z_1,z_2,z_3}\\
        &\quad+\frac{1}{N(N-1)}\sum_{\{z_i\}_{i=1}^2\in[N]^2_\mathrm{dist}}\delta_{b_1,b_2}(1-\delta_{b_2,b_3})\ketbra{\tau(z_1,z_1,z_2)}{z_1,z_1,z_2}\\
        &\quad+\frac{1}{N(N-1)}\sum_{\{z_i\}_{i=1}^2\in[N]^2_\mathrm{dist}}\delta_{b_2,b_3}(1-\delta_{b_1,b_2})\ketbra{\tau(z_1,z_2,z_2)}{z_1,z_2,z_2}\\
        &\quad+\frac{1}{N(N-1)}\sum_{\{z_i\}_{i=1}^2\in[N]^2_\mathrm{dist}}\delta_{b_1,b_3}(1-\delta_{b_2,b_3})\ketbra{\tau(z_1,z_2,z_1)}{z_1,z_2,z_1}\\
        &\quad+\frac{1}{N}\sum_{z\in[N]}\delta_{b_1,b_2}\delta_{b_2,b_3}\ketbra{z,z,z}
    \end{split}
\end{equation}
for any $\{b_1,b_2,b_3\}\in[K]^3$, $a\in[N/K]$, and $\tau\in S_3$. 
\end{proof}
The total variance is dominated by $\mathrm{Var}(\hat{o}_\mathrm{o.d.})$. For arbitrary $\rho$, this is upper bounded by
\begin{equation}
    \mathrm{Var}(\hat{o}_\mathrm{d.})+\mathrm{Var}(\hat{o}_\mathrm{o.d.})\leq O\left(\tr(\hat{O}^2)\right) + O\left(NK^{-1}\chi(\hat{O})\right)
\end{equation}
with
\begin{equation}
        \chi(\hat{O}) =\max_{z \in [N]} \bra{z}\hat{O}_{\rm o.d.}^2\ket{z} = \max_{z \in [N]} \left[ \bra{z} \hat{O}^2 \ket{z} - \bra{z} \hat{O}\ket{z}^2 \right].
\end{equation}
We note that $\tr \left( \hat{O}^2 \right)$ is the sample complexity for the case of random Clifford gates. We may take an appropriate scaling of $K$ depending on $\chi(\hat{O})$ to make $NK^{-1}\chi(\hat{O})\sim \tr(\hat{O}^2)$. This directly provides a trade-off between sample complexity and classical postprocessing time. 
For example, when $\chi(\hat{O}) = O(nN^{-1})$, one can take $K = O(n)$ to make $NK^{-1}\chi(\hat{O})\sim \tr(\hat{O}^2)=O(1)$. In this case, the total sample complexity becomes a constant if $\tr(\hat{O}^2)=O(1)$, \textit{i.e.}, $\hat{O}$ is a low-rank observable.
We note that for a typical low-rank $\hat{O}$ whose eigenstates are sampled from the ensemble of Haar random states, $\chi(\hat{O})$ is given by $O(n/N)$ with high probability. Furthermore, if we apply the product of random single-qubit unitaries to the underlying state and the observable, then we can reduce the subsystem size to $K=O(1)$ while preserving the constant sample complexity as discussed in \textbf{Theorem}~\ref{thm:almost-all-certification} in the next subsection. This preconditioning process can be utilized to measure a wider class of observables. For example, $N\chi(\hat{O})$ is exponentially large for relatively simple low-rank observables such as $\hat{O}=\ketbra{\mathrm{GHZ}}$. In this case, we can measure them by applying the product of Hadamard gates to the underlying state and the observable to make $N\chi({H^{\otimes n}\hat{O}H^{\otimes n}})=O(1)$.

\subsubsection*{Certification and benchmarking using shadow tomography (proof of \textbf{Theorem 5} in the main text)}
In this subsection, we apply our shadow tomography method to quantum state certification and benchmarking. These tasks represent some of the most significant applications of shadow tomography utilizing global random unitary operators, such as random Clifford circuits. We provide two primary results. In \textbf{Theorem}~\ref{thm:almost-all-certification}, we establish that almost all pure quantum states can be certified with constant sample complexity. Furthermore, \textbf{Theorem}~\ref{thm:benchmarking-shadow} demonstrates that the fidelity of anti-concentrated states, prepared by a quantum device subject to random errors and generic noise, can be estimated using a constant number of samples.

\begin{theorem}[\textbf{Theorem 5} {in the main text} restated]\label{thm:almost-all-certification}
    Let $\rho$ be a fixed $n$-qubit state. With high probability, for a $n$-qubit pure state $\ket{\phi}$ sampled from the Haar random ensemble, the fidelity between $\rho$ and $\ket{\phi}$ can be estimated up to $\varepsilon$-additive error through $O(\log (n/\varepsilon))$-depth circuits with $O(1/\varepsilon^2)$ queries to $\rho$ and fixed local basis coefficients of $\ket{\phi}$. 
\end{theorem}
\begin{proof}
    First of all, let $U$ be $\bigotimes_{i=1}^n U_i$ where $\{U_i\}_{i=1}^n$ are sampled independently from an ensemble forming single qubit unitary $2$-design. We define $p_z$ for $z\in[N]$ as $\bra{z}U\rho U^\dag\ket{z}$. The fidelity between $\rho$ and $\ket{\phi}$ can be estimated using the estimators $\hat{o}_\mathrm{d.}$ and $\hat{o}_\mathrm{o.d.}$ with $\hat{O}=U^\dag\ketbra{\phi}U$. Due to the unitary invariance of the Haar measure, $U^\dag\ket{\phi}$ can be thought of as a Haar random state. We will use the measurement unitary operator introduced in \textbf{Theorem}~\ref{thm:shadow-circuit}, which can be implemented in $O(\log(n/\varepsilon))$-depth with $\varepsilon$-additive error in $\mathbb{E}[\hat{o}_\mathrm{o.d.}]$. We note that \textbf{Theorem}~\ref{thm:sample-complexity} introduces the following upper bounds on the variances of the estimators $\hat{o}_\mathrm{d.}$ and $\hat{o}_\mathrm{o.d.}$:
    \begin{equation}
        \begin{split}
            \mathrm{Var}(\hat{o}_\mathrm{d.}) &\leq 1\\
            \mathrm{Var}(\hat{o}_\mathrm{o.d.}) &\leq O\left(N \sum_{z\in[N]} p_z q_z \right)
        \end{split}
    \end{equation}
    with $q_z=\abs{\bra{z}U^\dag\ket{\phi}}^2$. This means that the total sample complexity is proportional to the linear cross entropy $N\sum_{z\in[N]}p_z q_z-1$. Next, let us consider the maximum value of $\{p_z\}$. Due to \textbf{Lemma}~\ref{thm:onsite-Haar-max-prob}, with high probability, $p_\mathrm{max}$ is upper bounded by $2^{-O(n)}$. Now, \textbf{Lemma}~\ref{thm:cross-entropy-bound} implies that the linear cross entropy is upper bounded by $O(1)$ with failure probability smaller than $e^{-2^{O(n)}}$. Combining altogether, with high probability, the sum of the variances is $O(1)$, and $O(1/\varepsilon^2)$ shadow estimators are sufficient to estimate the fidelity between $\rho$ and $\ket{\phi}$ up to additive error $\varepsilon$.
\end{proof}

If one does not want to apply the single-qubit Haar random unitary layer as in this proof, then the following corollary may be helpful.
\begin{corollary}\label{thm:almost-all-certification-without-Haar}
    Let $\rho$ be a fixed $n$-qubit state. With high probability, for a $n$-qubit pure state $\ket{\phi}$ sampled from the Haar random ensemble, the fidelity between $\rho$ and $\ket{\phi}$ can be estimated up to $\varepsilon$-additive error through $O(\log (n/\varepsilon))$-depth circuits with $O(n/\varepsilon^2)$ queries to $\rho$ and computational basis coefficients of $\ket{\phi}$. 
\end{corollary}
\begin{proof}
    This corollary can be proven using the proof of \textbf{Theorem}~\ref{thm:almost-all-certification} except for the fact that the maximum probability $p_\mathrm{max}$ is an unknown constant. Here, we set $\gamma$ in \textbf{Lemma}~\ref{thm:cross-entropy-bound} as $n$. This gives $\mathrm{Var}(\hat{o}_\mathrm{o.d.})\leq O(n)$ with failure probability smaller than $e^{-O(n)}$. Thus, with high probability, $O(n/\varepsilon^2)$ shadow estimators are sufficient to estimate the fidelity between $\rho$ and $\ket{\phi}$ up to an additive error $\varepsilon$.
\end{proof}

\begin{lemma}\label{thm:cross-entropy-bound}
    Let $q$ be the measurement probability distribution of an $n$-qubit pure state $\ket{\phi}$ sampled from the Haar random ensemble. Let $p$ be a probability distribution over $[N]$ with $N=2^n$. We define $p_\mathrm{max}$ as the maximum value of $p$. For $\gamma>1$, the linear cross entropy between $p$ and $q$, $N\sum_{z\in[N]}p_z q_z-1$, is grater than $\gamma$ with probability smaller than $e^{-(\gamma-\log(1+\gamma))/p_\mathrm{max}}$.
\end{lemma}
\begin{proof}
    Let us compute the probability of having the linear cross-entropy greater than $\gamma$. To this end, we first note that for $z_0\in[N]$, $q_{z_0}$ is given by 
    \begin{equation}
        q_{z_0} = \frac{\hat{X}_{z_0}}{\sum_{z\in[N]}\hat{X}_z},
    \end{equation}
    where $\{\hat{X}_z\}$ are independent random variables following the Porter-Thomas distribution. In addition, the Porter-Thomas distribution is log-concave. Then, due to \textbf{Theorem 2.8} of Ref.~\cite{Kumar1983}, $\{q_z\}$ are negatively associated random variables. Now, let us consider the exponential function $f(x)=e^x$. This is a positive increasing function. Thus, using \textbf{Property $\mathrm{P}_2$} of Ref.~\cite{Kumar1983}, we have
    \begin{equation}
        \mathbb{E} \left[ \prod_{z\in[N]} e^{t_* p_z q_z} \right] \leq \prod_{z\in[N]}\mathbb{E}\left[ e^{t_* p_z q_z} \right]
    \end{equation}
    for any $t_*\in T$ with $T=[0,Np_\mathrm{max}^{-1})$. This, together with the Chernoff bound, gives
    \begin{equation}
        \mathrm{Prob}\left(\sum_{z\in[N]}p_zq_z\geq \frac{\gamma+1}{N}\right) \leq \inf_{t\in T} \left( e^{-t(\gamma+1)/N}\prod_{z\in[N]}\mathbb{E}\left[ e^{t p_z q_z} \right]\right).
    \end{equation}
    For large enough $n$, $q_{z_0}$ follows the Porter-Thomas distribution and gives
    \begin{equation}
        \mathbb{E}\left[ e^{t p_{z_0} q_{z_0}} \right] = \frac{1}{1-tp_{z_0}/N}.
    \end{equation}
    If we set $t_*$ by $Np_\mathrm{max}^{-1}(1-(1+\gamma)^{-1})$, then we have
    \begin{equation}
        \begin{split}
            \mathrm{Prob}\left(\sum_{z\in[N]}p_zq_z\geq \frac{\gamma+1}{N}\right) 
            &\leq e^{-\gamma/p_\mathrm{max}}\prod_{z\in[N]}\frac{1}{1-(1-(1+\gamma)^{-1})p_z/p_\mathrm{max}}\\
            &\leq e^{-\gamma/p_\mathrm{max}}\prod_{z\in[N]}\frac{1}{[1-(1-(1+\gamma)^{-1})]^{p_z/p_\mathrm{max}}}\\
            &= e^{-(\gamma-\log(1+\gamma))/p_\mathrm{max}}.
        \end{split}
    \end{equation}
\end{proof}

\begin{lemma}\label{thm:onsite-Haar-max-prob}
    For any $n$-qubit state $\rho$, $p_\mathrm{max}=\max_{z\in[N]}\bra{z}U\rho U^\dag\ket{z}$ with $N=2^n$ and $U=\bigotimes_{i=1}^n U_i$ where $\{U_i\}$ are sampled from a unitary 2-design satisfies
    \begin{equation*}
        \mathrm{Prob}\left(p_\mathrm{max}\geq \frac{1}{N} + 2^{(0.6-\frac{1}{2}\log_2 3)n} \right) \leq \frac{2^n}{1+2^{1.2 n}}.
    \end{equation*}
\end{lemma}
\begin{proof}
    For $z\in[N]$, let $p_z$ be $\bra{z}U\rho U^\dag\ket{z}$. Let us first compute the expected value and variance of $p_z$. 
    \begin{equation}
        \mu = \mathbb{E}_U[p_z] = \mathbb{E}_U \tr(\rho U^\dag\ketbra{z}U ) = \frac{1}{N}
    \end{equation}
    and
    \begin{equation}
        \begin{split}
            \sigma^2 
            &= \mathbb{E}_U[p_z^2] - \mu^2 \\
            &= \mathbb{E}_U \tr(\rho^{\otimes 2} U^{\otimes 2,\dag}\ketbra{z}^{\otimes 2} U^{\otimes 2}) - \frac{1}{2^{2n}}\\
            &= \frac{1}{6^n}\tr(\rho^{\otimes 2} (I+\operatorname{SWAP})^{\otimes n}) - \frac{1}{2^{2n}}\\
            &\leq \frac{1}{3^n}.
        \end{split}
    \end{equation}
    Next, due to the union bound, we have
    \begin{equation}
        \mathrm{Prob}\left( p_\mathrm{max} \geq \mu + \lambda\sigma \right) \leq \sum_{z\in[N]}\mathrm{Prob}\left( p_z \geq \mu + \lambda\sigma \right)
    \end{equation}
    for any $\lambda>0$. Now, let us use the one-sided Chebyshev inequality to bound the right-hand side:
    \begin{equation}
        \mathrm{Prob}\left( p_z \geq \mu + \lambda\sigma \right) \leq \frac{1}{1+\lambda^2}.
    \end{equation}
    If we set $\lambda$ as $2^{-0.6n}$, then we get the stated inequality
    \begin{equation}
        \mathrm{Prob}\left( p_\mathrm{max} \geq \frac{1}{N} + 2^{(0.6-\frac{1}{2}\log_2 3)n} \right) \leq \frac{2^n}{1+2^{1.2 n}}.
    \end{equation}
\end{proof}

When prior knowledge of the underlying state is available, our shadow tomography protocol becomes applicable to a broader class of observables. For instance, if the underlying state is anti-concentrated, then the sample complexity becomes a constant. Additionally, when attempting to prepare a quantum state on a device subject to random errors and generic noise, we can benchmark the state with constant sample complexity, provided the state is anti-concentrated as follows.
\begin{theorem}\label{thm:benchmarking-shadow}
    For a $n$-qubit anti-concentrated pure state $\ket{\psi}$, if $\rho$ is a state generated from $\ket{\psi}$ by applying a random error as well as a noisy channel $\mathcal{C}$ that decreases the collision probability such as depolarization channels, \textit{i.e.}, $\rho = \mathcal{C}(\ketbra{\phi})$ with $\ket{\phi}=\sqrt{1-q}\ket{\psi}+\sqrt{q}\ket{\psi_\perp}$ and $\ket{\psi_\perp}\sim\mathcal{E}_\mathrm{Haar}$, then the fidelity between $\rho$ and $\ket{\psi}$ can be estimated up to $\varepsilon$-additive error through $O(\log(n/\varepsilon))$-depth circuits with $O(1/\varepsilon^2)$ queries to $\rho$ and computational basis coefficients of $\ket{\psi}$.
\end{theorem}
\begin{proof}
    We first prove that $\ket{\phi}$ is anti-concentrated. The fourth moments of $\abs{\braket{z}{\phi}}$ is given by
    \begin{equation}
        \begin{split}
            \abs{\braket{z}{\phi}}^4 
            &= (1-q)^2\abs{\braket{z}{\psi}}^4+q^2\abs{\braket{z}{\psi_\perp}}^4+q(q-1)(\braket{\psi_\perp}{z}^2\braket{z}{\psi}^2+\braket{\psi}{z}^2\braket{z}{\psi_\perp}^2+\abs{\braket{z}{\psi}\braket{z}{\psi_\perp}}^2)\\
            &\quad+ 2q(1-q)\abs{\braket{z}{\psi}\braket{z}{\psi_\perp}}^2+2\sqrt{q(q-1)}(\braket{\psi_\perp}{z}\braket{z}{\psi}+\braket{\psi}{z}\braket{z}{\psi_\perp})((1-q)\abs{\braket{z}{\psi}}^2+q\abs{\braket{z}{\psi_\perp}}^2).
        \end{split}
    \end{equation}
    Under averaging over $\psi_\perp\sim\mathcal{E}_\mathrm{Haar}$, it becomes
    \begin{equation}
        \begin{split}
            \mathbb{E}_{\psi_\perp\sim\mathcal{E}_\mathrm{Haar}}\left[\abs{\braket{z}{\phi}}^4\right] 
            &= (1-q)^2\abs{\braket{z}{\psi}}^4 + q^2\mathbb{E}_{\phi_\perp\sim\mathcal{E}_\mathrm{Haar}}\left[\abs{\braket{z}{\psi_\perp}}^4\right]+3q(1-q)\abs{\braket{z}{\psi}}^2\mathbb{E}_{\psi_\perp\sim\mathcal{E}_\mathrm{Haar}}\left[ \abs{\braket{z}{\psi_\perp}}^2 \right]\\
            &= (1-q)^2\abs{\braket{z}{\psi}}^4 + \frac{2 q^2}{N(N+1)} + \frac{3q(1-q)}{N}\abs{\braket{z}{\psi}}^2
        \end{split}
    \end{equation}
    with $N=2^n$. Since $\mathcal{C}(\cdot)$ decreases the collision probability, it is upper bounded by
    \begin{equation}
        \mathbb{E}_{\psi_\perp\sim\mathcal{E}_\mathrm{Haar}}\left[\sum_{z\in[N]}(\bra{z}\rho\ket{z})^2\right] \leq O(N^{-1}).
    \end{equation}
    As in the proof of \textbf{Theorem}~\ref{thm:almost-all-certification}, the sample complexity for estimating $\bra{\phi}\rho\ket{\phi}$ is proportional to $1/\varepsilon^2$ times the cross entropy between $\rho$ and $\ket{\phi}$ which is upper bounded by
    \begin{equation}
        N\sum_{z\in[N]}\abs{\braket{z}{\phi}}^2\bra{z}\rho\ket{z} - 1 \leq N\left(\sum_{z\in[N]}\abs{\braket{z}{\phi}}^4\right)^{1/2}\left(\sum_{z\in[N]}(\bra{z}\rho\ket{z})^2\right)^{1/2} - 1 \leq O(1).
    \end{equation}
\end{proof}

\subsubsection*{Measurement circuit for shadow tomography}
In this subsection, we discuss an efficient construction of $U_p$ for shadow tomography. We find that by setting $k=1$, $U_p$ can be constructed by using only two types of CNOT gates: $U_{\mathrm{CNOT},0}$ and $U_{\mathrm{CNOT},1}$, which are controlled on states $\ket{0}$ and $\ket{1}$, respectively. The entire circuit diagram is illustrated in \hyperlink{fig:shadow-circuit}{Supplementary Figure 5}. Briefly, similar to \textbf{Algorithm 1}, it first applies a copy circuit to replicate the subsystem into the half of the system. It then randomizes the qubits outside the subsystem by applying the CNOT gates controlled by the subsystem and targeting the external qubits with a probability of $1/2$. Finally, it randomizes the subsystem using CNOT gates with the roles of control and target qubits reversed. Below, we discuss this process in more detail.

\begin{theorem}\label{thm:shadow-circuit}
    A measurement circuit for $\hat{o}_\mathrm{o.d.}$ with $k=1$ can be implemented in $O(\log(n/\varepsilon))$ depth with $\varepsilon$ additive error in $\mathbb{E}[\hat{o}_\mathrm{o.d.}]$ if $\hat{O}$ is anti-concentrated, \textit{i.e.}, 
    \begin{equation*}
        \sum_{z\in[2^n]}(\bra{z}\hat{O}\ket{z})^2\leq \frac{C}{2^n}
    \end{equation*}
    for some constant $C$.
\end{theorem}
\begin{proof}
Let us first consider the copy circuit $U_c$. Without loss of generality, we assume that we get $0^n\in[2^n]$ as a measurement outcome and apply an on-site Haar random unitary operator $V$. Then, by averaging the unitary operator over the Haar random ensemble $\mathcal{E}(2)$, we get 
\begin{equation}
    \mathbb{E}_{V\sim\mathcal{E}(2)}[V^{\otimes 2}\ketbra{0^n}^{\otimes 2}V^{\otimes 2,\dag}] = \frac{1}{6}\sum_{\{b_1,b_2\}\in[2]^2}\sum_{\tau\in S_2}\ketbra{\tau(b_1,b_2)}{b_1,b_2}\otimes\ketbra{0^{n-1}}^{\otimes 2}.
\end{equation}
The role of the copy circuit is to replicate $(b_1,b_2)$ to the half of the system, \textit{i.e.},
\begin{equation}
    U_c^{\otimes 2}\ketbra{\tau(b_1,b_2)}{b_1,b_2}\otimes\ketbra{0^{n-1}}^{\otimes 2}U_c^{\otimes 2} = \ketbra{\tau(b_1,b_2)}{b_1,b_2}^{\otimes (n/2)}\otimes \ketbra{0^{n/2}}^{\otimes 2}.
\end{equation}

The second stage is to randomize all qubits except the first one. This can be done by applying $U_{\mathrm{CNOT},0}$ and $U_{\mathrm{CNOT},1}$ with the probability of $1/2$. To be more precise, if $U_{\mathrm{CNOT},0}$ and $U_{\mathrm{CNOT},1}$ are applied based on uniformly random binary variables $\hat{X}_0$ and $\hat{X}_1$, then for any $\{a_1,a_2\}\in[2]^2$, we have

\begin{equation}
    \begin{split}
        &\mathbb{E}_{\hat{X}_0,\hat{X}_1}\left[(U^{\hat{X}_0}_{\mathrm{CNOT},0})^{\otimes 2} (U^{\hat{X}_1}_{\mathrm{CNOT},1})^{\otimes 2}\ketbra{\tau(b_1,b_2)}{b_1,b_2}\otimes\ketbra{\tau(a_1,a_2)}{a_1,a_2}(U^{\hat{X}_1}_{\mathrm{CNOT},1})^{\otimes 2,\dag} (U^{\hat{X}_0}_{\mathrm{CNOT},0})^{\otimes 2,\dag}\right] \\
        &= \frac{1}{4} \sum_{\{z_1,z_2\}\in[2]^2}\ketbra{\tau(b_1,b_2)}{b_1,b_2}\otimes\ketbra{\tau(z_1\oplus a_1, z_2\oplus a_2)}{z_1\oplus a_1, z_2\oplus a_2}\\
        &= \frac{1}{4} \sum_{\{z_1,z_2\}\in[2]^2}\ketbra{\tau(b_1,b_2)}{b_1,b_2}\otimes\ketbra{\tau(z_1, z_2)}{z_1, z_2}
    \end{split}
\end{equation}
if $b_1$ and $b_2$ are distinct, and 
\begin{equation}
    \begin{split}
        &\mathbb{E}_{\hat{X}_0,\hat{X}_1}\left[(U^{\hat{X}_0}_{\mathrm{CNOT},0})^{\otimes 2} (U^{\hat{X}_1}_{\mathrm{CNOT},1})^{\otimes 2}\ketbra{\tau(b_1,b_2)}{b_1,b_2}\otimes\ketbra{\tau(a_1,a_1)}{a_1,a_1}(U^{\hat{X}_1}_{\mathrm{CNOT},1})^{\otimes 2,\dag} (U^{\hat{X}_0}_{\mathrm{CNOT},0})^{\otimes 2,\dag}\right] \\
        &= \frac{1}{2} \sum_{z\in[2]}\ketbra{\tau(b_1,b_2)}{b_1,b_2}\otimes\ketbra{z\oplus a_1, z\oplus a_1}\\
        &= \frac{1}{2} \sum_{z\in[2]}\ketbra{\tau(b_1,b_2)}{b_1,b_2}\otimes\ketbra{z, z}
    \end{split}
\end{equation}
if we have $b_1=b_2$. Here, `$\oplus$' represents the summation modulo two. 

Consequently, the second stage gives us 
\begin{equation}
    \begin{split}
        &\mathbb{E}\left[U_\mathrm{RCNOT}^{\otimes 2}\ketbra{\tau(b_1,b_2)}{b_1,b_2}^{\otimes(n/2)}\otimes\ketbra{0^{n/2}}^{\otimes 2}U_\mathrm{RCNOT}^{\otimes 2,\dag}\right]\\
        &=\frac{1}{2^{2(n-1)}}(1-\delta_{b_1,b_2})\sum_{\{a_1,a_2\}\in[2^{n-1}]^2}\ketbra{\tau(b_1,b_2)}{b_1,b_2}\otimes\ketbra{\tau(a_1,a_2)}{a_1,a_2}\\
        &\quad+\frac{1}{2^{n-1}}\delta_{b_1,b_2}\sum_{a\in[2^{n-1}]}\ketbra{\tau(b_1,b_2)}{b_1,b_2}\otimes\ketbra{a,a},
    \end{split}
\end{equation}
where $U_\mathrm{RCNOT}$ is the random CNOT circuit.

The final stage is devoted to randomize the first qubit. In this stage, $U_{\mathrm{CNOT},0}$ and $U_{\mathrm{CNOT},1}$ are randomly applied targeting the first qubit and conditioned on other qubits. We will literately apply these gates sequentially conditioned on the second qubit to the $(r+1)$-th qubit for some $r>1$ as shown in \hyperlink{fig:shadow-circuit}{Supplementary Figure 5}. We will denote $g^{(r)}_{b_1,b_2,\tau}$ as the matrix after the $r$-th iteration. We find that $g^{(r)}_{b_1,b_2,\tau}$ is explicitly given by
\begin{equation}
    \begin{split}
        g^{(s)}_{b_1,b_2,\tau} &= \frac{1}{2^{2n}}(1-\delta_{b_1,b_2})\sum_{\substack{\{b'_1,b'_2\}\in[2]^2\\\{c^{(1)}_1,c^{(1)}_2\}\in[2]^2_\mathrm{dist}\\\{a_1,a_2\}\in[2^{n-2}]^2}}\ketbra{\tau(b'_1,b'_2)}{b'_1,b'_2}\otimes\ketbra{\tau(c_1^{(1)},c_2^{(1)})}{c_1,c_2}\otimes\ketbra{\tau(a_1,a_2)}{a_1,a_2}\\
        &\quad+\frac{1}{2^{2n}}(1-\delta_{b_1,b_2})\sum_{\substack{\{b'_1,b'_2\}\in[2]^2\\c^{(1)}\in[2]\\\{c^{(2)}_1,c^{(2)}_2\}\in[2]^2_\mathrm{dist}\\\{a_1,a_2\}\in[2^{n-3}]^2}}\ketbra{\tau(b'_1,b'_2)}{b'_1,b'_2}\otimes\ketbra{c^{(1)},c^{(1)}}\otimes\ketbra{\tau(c^{(2)}_1,c^{(2)}_2)}{c^{(2)}_1,c^{(2)}_2}\otimes\ketbra{\tau(a_1,a_2)}{a_1,a_2}\\
        &\quad+\cdots\\
        &\quad+\frac{1}{2^{2n}}(1-\delta_{b_1,b_2})\sum_{\substack{\{b'_1,b'_2\}\in[2]^2\\\{c^{(1)},\cdots,c^{(s-1)}\}\in[2]^{s-1}\\\{c^{(s)}_1,c^{(s)}_2\}\in[2]^2_\mathrm{dist}\\\{a_1,a_2\}\in[2^{n-s-1}]^2}}\ketbra{\tau(b'_1,b'_2)}{b'_1,b'_2}\otimes\ketbra{c^{(1)},c^{(1)}}\otimes\cdots\otimes\ketbra{c^{(s-1)},c^{(s-1)}}\\
        &\qquad\qquad\qquad\qquad\qquad\qquad\qquad\qquad\otimes\ketbra{\tau(c^{(2)}_1,c^{(2)}_2)}{c^{(s)}_1,c^{(s)}_2}\otimes\ketbra{\tau(a_1,a_2)}{a_1,a_2}\\
        &\quad+\cdots\\
        &\quad+\frac{1}{2^{2n-1}}(1-\delta_{b_1,b_2})\sum_{\substack{\{b'_1,b'_2\}\in[2]^2_\mathrm{dist}\\\{c^{(1)},\cdots,c^{(r)}\}\in[2]^{r}\\\{a_1,a_2\}\in[2^{n-r-1}]^2}}\ketbra{\tau(b'_1,b'_2)}{b'_1,b'_2}\otimes\ketbra{c^{(1)},c^{(1)}}\otimes\cdots\otimes\ketbra{c^{(r)},c^{(r)}}\\
        &\qquad\qquad\qquad\qquad\qquad\qquad\qquad\qquad\otimes\ketbra{\tau(a_1,a_2)}{a_1,a_2}\\
        &\quad+\frac{1}{2^n}\delta_{b_1,b_2}\sum_{z\in[2^n]}\ketbra{z,z}.
    \end{split}
\end{equation}

Next, let us compare our implementation with the one generated by the exact random injective map:
\begin{equation}
    \begin{split}
        g_{b_1,b_2,\tau} 
        &= \mathbb{E}_p\left[U_p^{\otimes 2}\ketbra{\tau(b_1,b_2)}{b_1,b_2}\otimes\ketbra{0^{n-1}}^{\otimes 2}U_p^{\otimes 2,\dag}\right]\\
        &= \frac{1}{2^n(2^n-1)}(1-\delta_{b_1,b_2})\sum_{\{z_1,z_2\}\in[2^n]^2_\mathrm{dist}}\ketbra{\tau(z_1,z_2)}{z_1,z_2}+\frac{1}{2^n}\delta_{b_1,b_2}\sum_{z\in[2^n]}\ketbra{z,z}.
    \end{split}
\end{equation}
The difference between $g^{(r)}_{b_1,b_2,\tau}$ and $g_{b_1,b_2,\tau}$ is given by
\begin{equation}
    \begin{split}
        \Delta g_{b_1,b_2,\tau} 
        &= g^{(r)}_{b_1,b_2,\tau} - g_{b_1,b_2,\tau}\\
        &= -\frac{1}{2^{2n}(2^n-1)}(1-\delta_{b_1,b_2})\sum_{\{z_1,z_2\}\in[2^n]^2_\mathrm{dist}}\ketbra{\tau(z_1,z_2)}{z_1,z_2}\\
        &\quad+\frac{1}{2^{2n}}(1-\delta_{b_1,b_2})\left(\sum_{\substack{z_1,z_2\in[2^n]^2_\mathrm{dist}\\(z_1)_1\neq (z_2)_1 \\(z_1)_2= (z_2)_2,\cdots,(z_1)_{r+1}=(z_2)_{r+1}}}-\sum_{\substack{z_1,z_2\in[2^n]^2_\mathrm{dist}\\(z_1)_1= (z_2)_1,\cdots,(z_1)_{r+1}=(z_2)_{r+1}}}\right)\ketbra{\tau(z_1,z_2)}{z_1,z_2}.
    \end{split}
\end{equation}
Let us study how this difference affects the expectation value of $\hat{O}_\mathrm{o.d.}$. To this end, we compute the average of $\tr({\hat{O}_\mathrm{o.d.}} U^\dag \ketbra{z} U)$ over measurement circuits generated by our circuit implementation $\mathcal{E}'$:
\begin{equation}
    \begin{split}
        &\mathbb{E}_{U\sim\mathcal{E}',z\sim\bra{z}U\rho U^\dag\ket{z}}\left[\tr({\hat{O}_\mathrm{o.d.}} U^\dag \ketbra{z} U)\right]\\
        &=\frac{1}{3(2^n-1)}\tr(\rho {\hat{O}_\mathrm{o.d.}}) + \frac{2^n}{6}\sum_{\{b_1,b_2\}\in[2]^2}\sum_{\tau\in S_2}\tr(\rho\otimes {\hat{O}_\mathrm{o.d.}}\Delta g_{b_1,b_2,\tau}).
    \end{split}
\end{equation}
Therefore, the bias caused by using the random CNOT circuit is precisely given by
\begin{equation}
    \begin{split}
        \Delta \langle \hat{O}_\mathrm{o.d.}\rangle 
        &= \frac{2^n(2^n-1)}{2}\sum_{\{b_1,b_2\}\in[2]^2}\sum_{\tau\in S_2}\tr(\rho\otimes {\hat{O}_\mathrm{o.d.}}\Delta g_{b_1,b_2,\tau})\\
        &= \frac{2^n-1}{2^{n+1}}\sum_{\{b_1,b_2\}\in[2]^2_\mathrm{dist}}\sum_{\tau\in S_2}\left(\sum_{\substack{\{z_1,z_2\}\in[2^n]^2_\mathrm{dist}\\(z_1)_1\neq (z_2)_1 \\(z_1)_2= (z_2)_2,\cdots,(z_1)_{r+1}=(z_2)_{r+1}}}-\sum_{\substack{\{z_1,z_2\}\in[2^n]^2_\mathrm{dist}\\(z_1)_1= (z_2)_1,\cdots,(z_1)_{r+1}=(z_2)_{r+1}}}\right)\\
        &\qquad\qquad\qquad\qquad\qquad\qquad\qquad\qquad\qquad\qquad\qquad\qquad\tr(\rho\otimes {\hat{O}_\mathrm{o.d.}}\ketbra{\tau(z_1,z_2)}{z_1,z_2})\\
        &\quad-\frac{1}{2^{n+1}}\sum_{\{b_1,b_2\}\in[2]^2_\mathrm{dist}}\sum_{\tau\in S_2}\sum_{\{z_1,z_2\}\in[2^n]^2_\mathrm{dist}}\tr(\rho\otimes {\hat{O}_\mathrm{o.d.}}\ketbra{\tau(z_1,z_2)}{z_1,z_2})\\
        &=\frac{2^n-1}{2^n}\left(\sum_{\substack{\{z_1,z_2\}\in[2^n]^2\\(z_1)_2= (z_2)_2,\cdots,(z_1)_{r+1}=(z_2)_{r+1}}}-2\sum_{\substack{\{z_1,z_2\}\in[2^n]^2\\(z_1)_1= (z_2)_1,\cdots,(z_1)_{r+1}=(z_2)_{r+1}}}\right)\tr(\rho\otimes {\hat{O}_\mathrm{o.d.}}\ketbra{z_2,z_1}{z_1,z_2})\\
        &\quad-\frac{1}{2^n}\tr(\rho \hat{O}_\mathrm{o.d.}).
    \end{split}
\end{equation}
This bias is dominated by the first term for large $n$. Now, let us elaborate on when the first term becomes small. To this end, we consider eigen decompositions of $\rho$ and $\hat{O}$:
{
\begin{equation}
    \rho = \sum_\lambda \lambda \ketbra{\psi_\lambda}
\end{equation}
and
\begin{equation}
    \hat{O} = \sum_\mu \mu \ketbra{\phi_\mu}.
\end{equation}
We first compute the error term associated with $\ket{\psi_\lambda}$ and $\ket{\phi_\mu}$. We then consider the total error. When take the eigenstates with certain eigenvalues $\lambda$ and $\mu$, for simplicity, we drop the subscripts of $\ket{\psi_\lambda}$ and $\ket{\phi_\mu}$.} Let Schmidt decompositions of $\ket{\psi}$ and $\ket{\phi}$ with a subsystem $A$ be 
\begin{equation}
    \ket{\psi} = \sum_i s_i \ket{\psi_i}_A \ket{\psi_i}_{A^c}
\end{equation}
and
\begin{equation}
    \ket{\phi} = \sum_i s'_i \ket{\phi_i}_A \ket{\phi_i}_{A^c}.
\end{equation}
If we set $A$ as the first $(r+1)$-qubits, then we have
\begin{equation}
    \begin{split}
        &\sum_{\substack{\{z_1,z_2\}\in[2^n]^2\\(z_1)_1= (z_2)_1,\cdots,(z_1)_{r+1}=(z_2)_{r+1}}}\tr(\rho\otimes {\hat{O}_\mathrm{o.d.}}\ketbra{z_2,z_1}{z_1,z_2})\\
        &=\sum_{\substack{z_1,z_2\in[2^n]^2 \\(z_1)_1= (z_2)_1,\cdots,(z_1)_{r+1}=(z_2)_{r+1}}}\left[\braket{z_1}{\psi}\braket{\psi}{z_2}\braket{z_2}{\phi}\braket{\phi}{z_1} - \braket{z_1}{\psi}\braket{\psi}{z_2}\bra{z_2}\mathrm{diag}(\ketbra{\phi})\ket{z_1}\right]\\
        &=\sum_{\substack{a_1,a_2\in[2^{n-|A|}]^2\\b\in[2^{|A|}]}}\braket{b,a_1}{\psi}\braket{\psi}{b,a_2}\braket{b,a_2}{\phi}\braket{\phi}{b,a_1}-\sum_{z\in[2^n]}\abs{\braket{z}{\psi}}^2\abs{\braket{z}{\phi}}^2\\
        &=\sum_{\substack{a_1,a_2\in[2^{n-|A|}]^2\\b\in[2^{|A|}]}}\sum_{i,j,k,l}s_is_js'_ks'_l\braket{b}{\psi_i}_A\braket{a_1}{\psi_i}_{A^c}\braket{\psi_j}{b}_A\braket{\psi_j}{a_2}_{A^c}\braket{b}{\phi_k}_A\braket{a_2}{\phi_k}_{A^c}\braket{\phi_l}{b}_A\braket{\phi_l}{a_1}_{A^c}\\
        &\quad-\sum_{z\in[2^n]}\abs{\braket{z}{\psi}}^2\abs{\braket{z}{\phi}}^2\\
        &=\sum_{b\in[2^{|A|}]}\sum_{i,j,k,l}s_is_js'_ks'_l\braket{b}{\psi_i}_A\braket{\psi_j}{b}_A\braket{b}{\phi_k}_A\braket{\phi_l}{b}_A\braket{\phi_l}{\psi_i}_{A^c}\braket{\psi_j}{\phi_k}_{A^c}-\sum_{z\in[2^n]}\abs{\braket{z}{\psi}}^2\abs{\braket{z}{\phi}}^2.
    \end{split}
\end{equation}
Furthermore, if we set $B$ as qubits in $[2,r+1]$, then the error term $\Delta\langle \hat{O}_\mathrm{o.d.} \rangle_{\psi,\phi}$ associated with $\ket{\psi}$ and $\ket{\phi}$ is given by
\begin{equation}
    \begin{split}
        \Delta\langle \hat{O}_\mathrm{o.d.} \rangle_{\psi,\phi}
        &= \frac{2^n-1}{2^n}\sum_{b\in[2^{|B|}]}\sum_{i,j,k,l}s_is_js'_ks'_l\braket{b}{\psi_i}_B\braket{\psi_j}{b}_B\braket{b}{\phi_k}_B\braket{\phi_l}{b}_B\braket{\phi_l}{\psi_i}_{B^c}\braket{\psi_j}{\phi_k}_{B^c}\\
        &\quad-\frac{2^n-1}{2^{n-1}}\sum_{b\in[2^{|A|}]}\sum_{i,j,k,l}s_is_js'_ks'_l\braket{b}{\psi_i}_A\braket{\psi_j}{b}_A\braket{b}{\phi_k}_A\braket{\phi_l}{b}_A\braket{\phi_l}{\psi_i}_{A^c}\braket{\psi_j}{\phi_k}_{A^c}\\
        &\quad- \frac{2^n-2}{2^n}\sum_{z\in[2^n]}\abs{\braket{z}{\psi}}^2\abs{\braket{z}{\phi}}^2-\frac{1}{2^n}\abs{\braket{\psi}{\phi}}^2.
    \end{split}
\end{equation}
The first term of the right-hand side is upper-bounded by
\begin{equation}
    \begin{split}
        &\sum_{b\in[2^{|B|}]}\sum_{i,j,k,l}s_is_js'_ks'_l\braket{b}{\psi_i}_B\braket{\psi_j}{b}_B\braket{b}{\phi_k}_B\braket{\phi_l}{b}_B\braket{\phi_l}{\psi_i}_{B^c}\braket{\psi_j}{\phi_k}_{B^c}\\
        &=\sum_{b\in[2^{|B|}]}\abs{\sum_{i,l}s_i s'_l \braket{b}{\psi_i}_B \braket{\phi_l}{b}_B \braket{\phi_l}{\psi_i}_{B^c} }^2\\
        &=\sum_{b\in[2^{|B|}]}\abs{\left(\sum_l s'_l \braket{b}{\phi_l}_B\ket{\phi_l}_{B^c}\right)^\dag\left(\sum_i s_i \braket{b}{\psi_i}_B\ket{\psi_i}_{B^c}\right)}^2\\
        &\leq \sum_{b\in[2^{|B|}]} \abs{\sum_l s'_l \braket{b}{\phi_l}_B\ket{\phi_l}_{B^c}}^2 \abs{\sum_i s_i \braket{b}{\psi_i}_B\ket{\psi_i}_{B^c}}^2 \\
        &= \sum_{b\in[2^{|B|}]} \bra{b} \rho_{\phi,B} \ket{b}\bra{b} \rho_{\psi,B} \ket{b}
    \end{split}
\end{equation}
with
\begin{equation}
    \rho_{\psi,B} = \tr_{B^c}(\ketbra{\psi}) = \sum_i s_i^2 \ketbra{\psi_i}_B
\end{equation}
and
\begin{equation}
    \rho_{\phi,B} = \tr_{B^c}(\ketbra{\phi}) = \sum_l (s'_l)^2 \ketbra{\phi_l}_B.
\end{equation}
The same applies to the second term, which gives
\begin{equation}
    \begin{split}
        \Delta\langle \hat{O}_\mathrm{o.d.} \rangle_{\psi,\phi} 
        &\leq \frac{2^n-1}{2^n} \sum_{b\in[2^{|B|}]}\bra{b} \rho_{\phi,B} \ket{b}\bra{b} \rho_{\psi,B} \ket{b} + \frac{2^n-1}{2^{n-1}} \sum_{b\in[2^{|A|}]} \bra{b} \rho_{\phi,A} \ket{b}\bra{b} \rho_{\psi,A} \ket{b} \\
        &\quad + \frac{2^n-2}{2^n}\sum_{z\in[2^n]}\abs{\braket{z}{\psi}}^2\abs{\braket{z}{\phi}}^2 + \frac{1}{2^n}\abs{\braket{\psi}{\phi}}^2.
    \end{split}
\end{equation}
The first two terms in this bound decay exponentially as a function of $r$ when reduced density matrices of $\ket{\psi}$ and $\ket{\phi}$ are both anti-concentrated, or more weakly, have exponentially decaying collision probabilities. We note that the eigenstates of a reduced density matrix of a typical state are all anti-concentrated. The third term can be estimated independently from $\hat{\rho}_\mathrm{d.}$ and is exponentially vanishing for typical $\ket{\psi}$. We neglect the last term as it also vanishes exponentially. Using this, we can also upper-bound the total error as follows.
\begin{equation}
    \begin{split}
        \Delta\langle \hat{O}_\mathrm{o.d.} \rangle
        &\leq \frac{2^n-1}{2^n} \sum_{b\in[2^{|B|}]}\bra{b} \rho_{B} \ket{b}\bra{b} \hat{O}_{B} \ket{b} + \frac{2^n-1}{2^{n-1}} \sum_{b\in[2^{|A|}]} \bra{b} \rho_{A} \ket{b}\bra{b} \hat{O}_{A} \ket{b} \\
        &\quad + \frac{2^n-2}{2^n}\sum_{z\in[2^n]} \bra{b}\rho\ket{b}\bra{b}\hat{O}\ket{b} + \frac{1}{2^n}\tr(\rho\hat{O})
    \end{split}
\end{equation}
with $\rho_B = \tr_{B^c}(\rho)$ and $\hat{O}_B = \tr_{B^c}(\hat{O})$. Thus, if we have 
\begin{equation}
    \begin{split}
        \sum_{b\in[2^{|A|}]}\bra{b}\rho_A\ket{b}\bra{b}\hat{O}_A\ket{b} &\leq 2^{-O(|A|)}\\
        \sum_{b\in[2^{|B|}]}\bra{b}\rho_B\ket{b}\bra{b}\hat{O}_B\ket{b} &\leq 2^{-O(|B|)}\\
        \sum_{b\in[2^n]}\bra{b}\rho\ket{b}\bra{b}\hat{O}\ket{b} &\leq 2^{-O(n)},
    \end{split}
\end{equation}
then the total error decays exponentially as a function of $r$. Since the total circuit depth is given by $O(\log n)+r$, for a fixed error $\varepsilon>0$, we can achieve $O(\log (n/\varepsilon))$ depth. Due to \textbf{Lemma}~\ref{thm:subspace-anti-concentration}, if $\hat{O}$ is anti-concentrated, then $\hat{O}_A$ and $\hat{O}_B$ are also anti-concentrated. These give
\begin{equation}
    \begin{split}
        \sum_{b\in[2^{|A|}]}\bra{b}\rho_A\ket{b}\bra{b}\hat{O}_A\ket{b} &\leq \left(\sum_{b\in[2^{|A|}]} (\bra{b}\hat{O}_A\ket{b})^2\right)^{1/2} \leq 2^{-O(|A|)}\\
        \sum_{b\in[2^{|B|}]}\bra{b}\rho_B\ket{b}\bra{b}\hat{O}_B\ket{b} &\leq \left(\sum_{b\in[2^{|B|}]} (\bra{b}\hat{O}_B\ket{b})^2 \right) \leq 2^{-O(|B|)}\\
        \sum_{b\in[2^n]}\bra{b}\rho\ket{b}\bra{b}\hat{O}\ket{b} &\leq \left(\sum_{b\in[2^n]}\bra{b}\hat{O}\ket{b}\right)^{1/2} \leq 2^{-O(n)}.
    \end{split}
\end{equation}
\end{proof}

\begin{lemma}\label{thm:subspace-anti-concentration}
    For a $n$-qubit state $\ket{\psi}$, if $\ket{\psi}$ is anti-concentrated, then for any subsystem $A$, $\rho_A=\tr_{A^c}(\ketbra{\psi})$ is also anti-concentrated.
\end{lemma}
\begin{proof}
    For $a\in[2^{|A|}]$ and $b\in[2^{n-|A|}]$, let $p_{a,b}$ be the computational basis measurement probability distribution of $\ket{\psi}$, \textit{i.e.}, $p_{a,b}=\abs{\braket{a,b}{\psi}}^2$. Since $\ket{\psi}$ is anti-concentrated, there exists $C>0$ such that 
    \begin{equation}
        \sum_{a\in[2^{|A|}],b\in[2^{n-|A|}]}p_{a,b}^2 \leq \frac{C}{2^n}.
    \end{equation}
    Let $q_a$ be the measurement probability distribution of $\rho_A$. It is given by
    \begin{equation}
        q_a = \sum_{b\in[2^{n-|A|}]}p_{a,b}.
    \end{equation}
    Now, let us upper bound the inverse participation ratio of $q_a$:
    \begin{equation}
        \begin{split}
            \sum_{a\in[2^{|A|}]}q_a^2 
            &\leq \sum_{a\in[2^{|A|}]} 2^{n-|A|} \sum_{b\in[2^{n-|A|}]}p_{a,b}^2\\
            &\leq 2^{n-|A|}\times\frac{C}{2^n}\\
            &= \frac{C}{2^{|A|}}.
        \end{split}
    \end{equation}
    Therefore, if $\ket{\psi}$ is anti-concentrated, then $\rho_A$ is also anti-concentrated.
\end{proof}

\newpage
\subsection*{8. Application of approximate state design with a fixed $\epsilon$}
An example requiring approximate state design is the average fidelity estimation of a quantum channel $\Lambda$ defined as
\begin{equation}
    F_{\rm avg}(\Lambda) = \mathbb{E}_{\psi\sim{\rm Haar}} \left[ \bra{\psi}\Lambda(\ketbra{\psi})\ket{\psi}\right].
\end{equation}
We can approximate $F_{\rm avg}(\Lambda)$ by replacing Haar random states to an ensemble $\mathcal{E}$ forming an $\epsilon$-approximate state 2-design as
\begin{equation}
    \tilde{F}_\mathrm{avg}(\Lambda) = \mathbb{E}_{\psi\sim\mathcal{E}} \left[ \bra{\psi}\Lambda(\ketbra{\psi})\ket{\psi}\right].
\end{equation}
This replacement gives $O(\epsilon)$ additive error as follows:
\begin{equation}
    \begin{split}
        \tilde{F}_\mathrm{avg}(\Lambda) 
        &=\mathbb{E}_{\psi\sim\mathcal{E}}\tr(\Lambda(\ketbra{\psi})\ketbra{\psi})\\
        &=\mathbb{E}_{\psi\sim\mathcal{E}}\tr(\operatorname{SWAP}\ketbra{\psi}\otimes \Lambda(\ketbra{\psi}))\\
        &=\mathbb{E}_{\psi\sim\mathcal{E}}\tr(\operatorname{SWAP}(I\otimes U)(\ketbra{\psi}\otimes\ketbra{\psi}\otimes\ketbra{0}_\mathsf{A})(I\otimes U)^\dag)\\
        &=(1-\epsilon)\tr(\operatorname{SWAP}(I\otimes U)(\rho^{(2)}_\mathrm{Haar}\otimes\ketbra{0}_\mathsf{A})(I\otimes U)^\dagger)+\epsilon\tr(\operatorname{SWAP}(I\otimes U)(\rho^{(2)}_\mathrm{err}\otimes\ketbra{0}_\mathsf{A})(I\otimes U)^\dagger)\\
        &=F_\mathrm{avg}(\Lambda)+O(\epsilon),
    \end{split}
\end{equation}
with $\rho^{(2)}_\mathrm{Haar}=\mathbb{E}_{\psi\sim\mathrm{Haar}}[\ketbra{\psi}^{\otimes 2}]$ and $\Lambda(\ketbra{\psi})=\tr_\mathsf{A}(U\ketbra{\psi}\otimes\ketbra{0}_\mathsf{A}U^\dag)$. Here, we use the fact that $\mathbb{E}_{\psi\sim\mathcal{E}}[\ketbra{\psi}^{\otimes 2}]=(1-\epsilon)\rho^{(2)}_\mathrm{Haar}+\epsilon \rho^{(2)}_\mathrm{err}$ for some two-copy density matrix $\rho^{(2)}_\mathrm{err}$. Additionally, the estimation of $\hat{F}_\mathrm{avg}(\Lambda)$ requires a constant sample complexity since we have
\begin{equation}
    \tilde{F}_\mathrm{avg}(\Lambda) = \mathbb{E}_{U\ket{0}\sim\mathcal{E},z\sim \bra{z}U^\dag\Lambda(U\ketbra{0}U^\dag)U\ket{z}}[\delta_{z,0}]
\end{equation}
and
\begin{equation}
    \begin{split}
        \mathrm{Var}(\delta_{z,0})
        &=\mathbb{E}_{U\ket{0}\sim\mathcal{E},z\sim \bra{z}U^\dag\Lambda(U\ketbra{0}U^\dag)U\ket{z}}[\delta_{z,0}^2]-(\mathbb{E}_{U\ket{0}\sim\mathcal{E},z\sim \bra{z}U^\dag\Lambda(U\ketbra{0}U^\dag)U\ket{z}}[\delta_{z,0}])^2\\
        &= \tilde{F}_\mathrm{avg}(\Lambda) - (\tilde{F}_\mathrm{avg}(\Lambda))^2\\
        &\leq 1.
    \end{split}
\end{equation}
In general, we can approximately estimate the observable with the following form
\begin{equation}
    \mathbb{E}_{\psi\sim\mathrm{Haar}}[g(\{\bra{\psi}A_m\ket{\psi}\})]
\end{equation}
with $O(\epsilon)$ additive error, where $g(x_1,\cdots,x_t)$ is a polynomial function of degree $t$. We note that these tasks can be done using pseudorandom states~\cite{Ji2018} only when the observer's computation time is limited. This constraint, however, is artificial in these cases as it has no connection with the physical meanings of observables, including fidelity. Even with this constraint, they are not as efficient as approximate state designs for these tasks.

\newpage
\section*{Supplementary References}
\bibliography{Ref}